\documentclass[runningheads]{llncs}
\usepackage{float}
\usepackage[T1]{fontenc}
\usepackage[utf8]{inputenc}
\usepackage[normalem]{ulem}
\usepackage[inline]{enumitem}
\usepackage[]{outlines}
\usepackage{enumitem}
\usepackage[dvipsnames]{xcolor}
\usepackage{pifont}
\usepackage{amsmath,amssymb,comment,scalerel}
\usepackage{array}
\usepackage{ragged2e}
\usepackage{makecell}
\usepackage{svg}
\usepackage{tikz}
\usetikzlibrary{arrows.meta, positioning}
\usepackage{bussproofs}
\usepackage{listings}
\usepackage{lipsum}
\usepackage{stmaryrd}
\usepackage[linesnumbered,noend]{algorithm2e} 
\SetKwInput{KwIn}{In}
\SetKwInput{KwOut}{Out}
\SetArgSty{upshape}
\SetKw{KwGoTo}{go to}
\SetKwProg{Fn}{def}{\string:}{}

\SetKwProg{KwMatch}{match}{:}{}

\usepackage[bookmarksnumbered,unicode]{hyperref}
\hypersetup{                    
  colorlinks,
  pdfborder={0 0 0},
  pdfmenubar=false,
  pdfpagemode=UseNone,
  pdfnonfullscreenpagemode=UseNone,
  bookmarksopen=true,
  breaklinks=true,
  linkcolor={green!50!black},
  citecolor={red!50!black},
  urlcolor={blue!50!black}
}
\usepackage[hide]{comments}
\Commenter{LP}
\Commenter{JP}
\Commenter{CB}
\Commenter{VR}

\usepackage[square,comma,numbers,sort&compress]{natbib}
\usepackage[english]{babel}
\usepackage{xspace}
\usepackage{tikz}
\usetikzlibrary{arrows.meta,decorations.pathreplacing,scopes,calc}
\tikzset{>=Latex[]}
\usepackage{booktabs}
\usepackage{caption}
\usepackage{subcaption}
\usepackage{mdframed}
\usepackage[most]{tcolorbox}
\usepackage{rotating}
\usepackage{cleveref}

\let\llncssubparagraph\subparagraph
\let\subparagraph\paragraph
\usepackage[compact]{titlesec}
\let\subparagraph\llncssubparagraph
\titlespacing*{\section}{0pt}{*2.0}{*2.0}   
\titlespacing*{\subsection}{0pt}{*1.5}{*1.0} 
\titlespacing*{\subsubsection}{0pt}{*0.6}{*0.4}

\newif\ifshortversion
\shortversionfalse

\newif\ifanonymous

\newcommand{\xspacemm}{\ifmmode\else\xspace\fi}
\newcommand{\mathnoun}[1]{\ensuremath{#1}\xspacemm}
\newcommand{\newmathnoun}[2]{\newcommand{#1}{\mathnoun{#2}}}

\newmathnoun{\FF}{\mathbb{F}}
\newmathnoun{\BB}{\mathbb{B}}
\newmathnoun{\ZZ}{\mathbb{Z}}
\newmathnoun{\ZZNN}{\mathbb{Z}_{\geq0}}
\newmathnoun{\PosZZ}{\mathbb{Z}_{\infty}^{+}}
\newmathnoun{\QQ}{\mathbb{Q}}
\newmathnoun{\NN}{\mathbb{N}}
\newmathnoun{\NNOne}{\mathbb{N}_{> 0}}
\newmathnoun{\C}{\mathcal{C}}
\newmathnoun{\EE}{\mathbb{E}}
\newmathnoun{\GG}{\mathbb{G}}
\newmathnoun{\sP}{\mathcal{P}}
\newmathnoun{\KK}{\mathbb{K}}

\newmathnoun{\BV}{\mathtt{BV}}
\newmathnoun{\BVN}{\mathtt{BV}_{N}}
\newmathnoun{\BVNoN}{\mathtt{BV}}
\newmathnoun{\BVB}{\mathtt{BV}_{B}}
\newmathnoun{\FFp}{\mathbb{F}_P}
\newmathnoun{\FFq}{\mathbb{F}_q}
\newmathnoun{\LC}{\mathsf{LC}}

\newmathnoun{\Expr}{\mathtt{Expr}}
\newmathnoun{\Bool}{\mathtt{Bool}}

\newcommand{\bv}{\ensuremath{\mathit{bv}}\xspace}
\newcommand{\fe}{\ensuremath{\mathit{e}}\xspace}
\newcommand{\nt}{\ensuremath{\mathit{t}}\xspace}

\newmathnoun{\ffbv}{\texttt{ff2bv}}

\newmathnoun{\bvff}{\texttt{bv2ff}}

\newmathnoun{\True}{\mathtt{True}}
\newmathnoun{\False}{\mathtt{False}}

\ifanonymous
  
\else
  
\fi

\newcommand{\tabCaptionSpace}{\vspace{1ex}}

\newcommand{\gbs}{Gr\"obner bases\xspace}

\newcommand{\mvZModSub}{\texttt{mvZModSub}\xspacemm}

\newcommand{\toNat}{\texttt{toNat}}

\newcommand{\bvToNat}{\texttt{bvToNat}}

\newcommand{\toBV}[2]{\texttt{toBV}(#1,#2)}
\newcommand{\toBVN}[1]{\texttt{toBV}(N,#1)}
\newcommand{\toBVOne}[1]{\texttt{toBV}(1,#1)}
\newcommand{\toBVNoN}{\texttt{toBV}\xspace}

\newcommand{\ZMod}{\texttt{ZMod}\, $P$ \xspacemm}

\newcommand{\injNat}{\texttt{injNat}\xspacemm}

\newcommand{\injBitVec}{\texttt{injBV}\xspacemm}
\newcommand{\injBitVecLeq}{\texttt{injBVLeqHyp}\xspacemm}

\newcommand{\calcBitWidth}{\texttt{clcBVWdth}\xspacemm}

\newcommand{\mvNat}{\texttt{distNat}\xspacemm}

\newcommand{\mvNatIte}{\texttt{distNatIte}\xspacemm}

\newcommand{\mvNatSub}{\texttt{distNatSub}\xspacemm}
\newcommand{\mvNatSubOverflow}{\texttt{distNatSubOvrflw}\xspacemm}

\newcommand{\mvBV}{\texttt{distBV}\xspacemm}
\newcommand{\mvBVIte}{\texttt{distBVIte}\xspacemm}

\newcommand{\mvBVSub}{\texttt{distBVSub}\xspacemm}
\newcommand{\mvMod}{\texttt{mvMod}\xspacemm}
\newcommand{\dropMod}{\texttt{dropMod}\xspacemm}
\newcommand{\addBds}{\texttt{addBds}\xspacemm}
\newcommand{\sqrBds}{\texttt{sqrBds}\xspacemm}

\newcommand{\mvBVMod}{\texttt{distBVMod}\xspacemm}
\newcommand{\BitModEq}{\texttt{BitModEq}\xspacemm}

\newcommand{\bvdecide}{\texttt{bv\_decide}\xspacemm}

\newcommand{\introMVar}{\texttt{introPVar}\xspacemm}
\newcommand{\rmvMVar}{\texttt{rmvPVar}\xspacemm}

\newcommand{\ineqAddMul}{\texttt{leqAddMul}\xspacemm}
\newcommand{\leSub}{\texttt{leqSub}\xspacemm}

\newcommand{\ineqMod}{\texttt{leqMod}\xspacemm}
\newcommand{\ineqHyp}{\texttt{ineqHyp}\xspacemm}
\newcommand{\ineqConst}{\texttt{ineqConst}\xspacemm}
\newcommand{\leZMod}{\texttt{leqZMod}\xspacemm}
\newcommand{\geNat}{\texttt{geNat}\xspacemm}
\newcommand{\leBV}{\texttt{leqBV}\xspacemm}
\newcommand{\ineqCases}{\texttt{ineqCases}\xspacemm}
\newcommand{\ineqMux}{\texttt{ineqXOR}\xspacemm}
\newcommand{\leqIf}{\texttt{leqIf}\xspacemm}

\newcommand{\decide}{\texttt{eval}\xspacemm}
\newcommand{\eval}{\texttt{eval}\xspacemm}
\newcommand{\OnlyVars}{\texttt{twoOrigVars}\xspacemm}

\newcommand{\HasVars}
{\texttt{origOnly}\xspacemm}

\newcommand{\HasSub}{\texttt{hasSub}\xspacemm}

\newcommand{\toBitVec}{\texttt{to\_}\BVNoN}
\newcommand{\toNN}{\texttt{to\_}\NN}

\newmathnoun{\ops}{\mathcal{O}}
\newmathnoun{\vars}{\mathcal{V}}
\newmathnoun{\expps}{\mathcal{E}}

\newmathnoun{\ogVars}{\mathtt{origVars}}
\newmathnoun{\pVars}{\mathtt{pVars}}

\begin{document}
\title{
  Automating Bitvector and Finite Field Equivalence Proofs in Lean
}
\titlerunning{
Automating Bitvector and Finite Field  Equivalence Proofs in Lean
}


\ifanonymous
  \author{}
  \institute{}
\else
  \author{%
    Elizaveta Pertseva\inst{1} \and
    Valentin Robert\inst{2} \and
    Clark Barrett\inst{1} \and
    James Parker\inst{2}
  }
  \institute{%
     Stanford University (\email{pertseva@stanford.edu})
     \and
     Galois
  }
  \authorrunning{%
    E. Pertseva et al.
  }
\fi

\maketitle
\vspace*{-1.1cm}
\begin{abstract}
%

Efforts to verify Zero-Knowledge Proof circuit encodings have highlighted the challenge of
proving the correctness of quantifier-free statements
that make use of both bitvector and finite field operations. Existing verification workflows are either manual or rely on SMT solvers, which scale poorly on some classes of problems for reasons that include difficulties with conversion operators and challenges reasoning about inequalities. To address these limitations, we present a novel Lean tactic \BitModEq that leverages range lemmas and case analysis to produce verified translations from finite fields to bitvectors. Our approach, combined with bit-blasting, outperforms state-of-the-art SMT solvers, solving 19\% more ZKP arithmetization benchmarks.

\end{abstract}
\vspace*{-.4cm}

\margincommentprimer

\section{Introduction}%
\label{sec:intro}

\JP{Abstract + Intro + Overview are frozen. Leave comments if necessary}Recent advances in Zero-Knowledge Proofs (ZKPs) have led to the utilization of ZKPs in critical applications such as cryptocurrencies~\cite{hopwood2016zcash}, post-quantum signature schemes~\cite{baum2023faest}, voting systems~\cite{miao2023secure}, proofs of vulnerabilities~\cite{cuellar2023cheesecloth,cuellar2025cheesecloth}, and image authentication~\cite{datta2025veritas}.  
ZKPs are cryptographic algorithms that enable one party (the prover) to convince another party (the verifier) that some secret data satisfies a property $\phi$ without revealing any information about the data. 

ZKP statements use the language of finite fields, where operations such as addition and multiplication are performed modulo a prime.  
However, real-world systems operate over concrete machine integers, using arithmetic modulo some power of two, also known as the theory of \emph{bitvectors}. 
As a result, ZKP statements often require arithmetizations which encode bitvector operations as finite field operations.  
The correctness of these arithmetizations is crucial for the soundness of ZKP systems, but constructing them can be error-prone. 

Prior work on verifying ZKP arithmetizations has followed two main approaches: manual and automatic.  
The manual approach typically works by embedding entire ZKP systems in interactive theorem provers (ITPs), as seen in formalizations of SP1 Hypercube and Cairo \cite{succinct2025, avigad2022verified}.  
While these efforts provide strong correctness guarantees, they require substantial human effort.  
Further, because ZKP optimizations evolve rapidly, arithmetizations frequently change, requiring significant portions of the verification effort to be revisited.

The automatic approach reduces verification to solving Satisfiability Modulo Theories (SMT) queries \cite{ozdemir2025bounded, stephens2025automated}.  
While this approach enables greater automation, it poses two fundamental challenges 
that limit what can be done automatically.
First, arithmetizations simultaneously involve bitvectors and finite fields, requiring the use of theory conversion operators that create performance problems for SMT solvers.
Second, SMT solvers contain a much larger trusted code base (TCB) than ITPs, which increases the probability of missed bugs.\footnote{Work on emitting proof certificates from SMT solvers is one approach for addressing this~\cite{BBC+23}, but proof support for the theory of finite fields remains ongoing research.}

To address these challenges, we present \BitModEq, a Lean tactic that translates finite field operations to bitvectors and solves the resulting goals via bit-blasting~\cite{bvdecide}.  
\BitModEq retains the guarantees of ITPs and outperforms existing automatic methods.  
Our approach is guided by two main insights.  

First, the majority of ZKP arithmetizations do not inherently rely on field-specific behavior (e.g., inverses or field division).  
Instead, they simulate bit-level operations inside the field, producing large numbers of inequality constraints that overwhelm the finite field reasoning engines.
To remedy this, we propose automatically translating field elements back to bitvectors during verification, enabling efficient reasoning over the original bit-level semantics. 

Our second insight is that the effectiveness of the translation primarily relies on range analysis.  
With tighter bounds, we can more frequently eliminate modulo operations, creating bitvectors of smaller widths.  
While Lean already has built-in lemmas for reasoning about equation ranges, they do not consider variable dependencies, which are crucial in bitvector operations simulating arithmetic.  
To remedy this, we design a Lean-specific range analysis that mixes Lean lemmas with targeted case reasoning.
Our main contributions are as follows: 
\begin{enumerate}
    \item A  Lean tactic \BitModEq that translates from finite fields to bitvectors. 
    \item A proof-of-concept bounded verification of two real-world ZKP encodings that increases the bounds achieved by previous techniques. 
    \item A Lean range analysis tactic for proving goals resulting from the translation.
\end{enumerate}

The remainder of the paper is organized as follows.  We first present a motivating example (Section~\ref{sec:motivate}) and relevant background (Section~\ref{sec:bg}).  We then describe our translation algorithm (Section~\ref{sec:translation}), range analysis tactic (Section~\ref{sec:range_analysis}), and Lean implementation (Section~\ref{sec:implementation}). We conclude with an evaluation (Section~\ref{sec:evaluation}), related work (Section~\ref{sec:relwork}), and a discussion of future work (Section~\ref{sec:discuss}).
\section{Motivating Example}
\label{sec:motivate}
\begin{table}[t]
\centering
\setlength{\tabcolsep}{3pt}
\begin{tabular}{cc||cc||c}
 
  $x_0$ &   $y_0$   & $x_0+y_0$  & $x_0\cdot y_0$   & $2^0 \cdot (x_0+y_0 - x_0\cdot y_0)$    \\
    \bottomrule 
  0 & 0 & 0 & 0 & 0 
  \\
  1 & 0& 1& 0 & 1
  \\ 0 &1  & 1 & 0 & 1  \\
  1 & 1 & 2 & 1 & 1
  \end{tabular}
  \tabCaptionSpace
  \caption{Evaluation table for the finite field polynomial encoding \eqref{eq:bvor_poly} of \texttt{bvor} (bitwise \texttt{OR}) in the Jolt zkVM for 1-bit inputs \cite{arun2024jolt}.}
  \label{fig:bit_split}
  \vspace{-0.3in}
\end{table}
Consider the polynomial shown in \eqref{eq:bvor_poly} over a finite field of order $P$, where $P$ is larger than $2^B$. The polynomial is an encoding of the \texttt{bvor} (bitvector \texttt{OR}) operation
used in the Jolt zkVM~\cite{arun2024jolt}.
In this encoding, $x_i$ represents the $i$-th bit of one operand, and $y_i$ represents the $i$-th bit of the other. We show all possible values for the polynomial and its subexpressions for the case $B=1$ in Table~\ref{fig:bit_split}.

\begin{equation}
\label{eq:bvor_poly}
\sum_{i=0}^{B-1} 2^i \bigl( x_i + y_i - x_i \cdot y_i \bigr)
\end{equation}
To verify that~\eqref{eq:bvor_poly} correctly encodes the \texttt{bvor} operation, we must encode the relationship between the bitvector operation and the finite field polynomial~\eqref{eq:bvor_poly}. Inspired by prior work~\cite{ozdemir2025bounded}, we rely on a predicate $\ffbv(\bv, \fe)$ (read as $\bv$ is the bitvector encoding of the finite field element $\fe)$ that encodes an equivalence between a bitvector $\bv$ of bit width $N$ and a finite field term $\fe$. While the exact definition of $\ffbv$ varies across systems~\cite{circ, arun2024jolt, succinct2025}, it most often includes a \emph{range check}. In Jolt, $\ffbv(\bv, \fe)$ is defined if and only if the value of $\fe$ is  smaller than the target bitvector size, $2^N$; otherwise, the operation is undefined.\footnote{Other systems may apply a modulo operation or return zero when  $\fe$ is too large.}  Let $\bv_1, \bv_2$ be bitvectors of bit width $B$, and let $x_i, y_i$ be finite field elements from field $P$, for $i\in\{0,\dots,B-1\}$. Our full verification query becomes:
\refstepcounter{equation}
\begin{prooftree}
  \AxiomC{$\bigwedge_{i \in \{0...B-1\}} \;  
 \ffbv(\bv_1[i],x_i)  
  \;\qquad\; 
  \bigwedge_{i \in \{0...B-1\}} \;  
  \ffbv( \bv_2[i],y_i)$}
 \RightLabel{\hspace{.7cm}(2)}
\alwaysSingleLine
  \UnaryInfC{$  
 \ffbv(\textbf{bvor} \; \bv_1 \;  \bv_2, \sum_{i=0}^{B-1} 2^i(x_i + y_i - x_i \cdot y_i))$} 
  \label{eq:corror1}
\end{prooftree}

One of the main focuses of this paper is the question of how to efficiently embed the $\ffbv$ operator. Previous SMT-based approaches are bottlenecked by the \ffbv embedding  for two main reasons. First, since SMT solvers do not support direct conversion from finite fields to bitvectors, the encoding must use bit splitting, which adds a large number of variables. Second, since current SMT theories of finite fields do not support inequalities directly \cite{hader2024smt}, the SMT encoding must use a finite field polynomial $f(\fe) = \prod_{i=0}^{2^N-1}(\fe-i)$ for range checks,\footnote{$f(\fe)$ is zero on and only on $\{0,\dots,2^N-1\}$ and thus fixes the possible values of $\fe$ to the interval $\{0,\dots,2^N-1\}$~\cite{ozdemir2025bounded}.} which slows down the solver.
The full SMT encoding of \ffbv is as follows:
\begin{equation}
\label{eq:ffbv_smt}
\begin{aligned}
\ffbv_{\mathtt{SMT}}(\bv, \fe) &
\iff\;
\prod_{i=0}^{2^N-1} (\fe-i) = 0
\;\land\;
\exists\: \fe_0 \cdots \fe_{N-1}. \,
\fe = \sum_{i=0}^{N-1} 2^i \fe_i \\
&\land\;
\bigwedge_{i=0}^{N-1} \fe_i(\fe_i - 1) = 0
\;\land\;
\bigwedge_{i=0}^{N-1}
\bv[i]
=
\texttt{ite}(\fe_i = 0,\,
\mathtt{0}_{[1]},\,
\mathtt{1}_{[1]}),
\end{aligned}
\end{equation}
where $c_{[B]}$ denotes a bitvector constant with value $c$ and bit width $B$, and $2^N$ is smaller than the order of the finite field, as is the case for Jolt.  A weaker SMT encoding can omit the range check polynomial, assuming instead that no overflow occurs during the conversion.  If this assumption is correct, bit splitting already constrains the values of $\fe$ appropriately. We refer to this as the \emph{weak encoding}, as it relies on an unproved assumption. However, even the weak version of \ffbv remains challenging for SMT solvers, as bit splitting is still present.

In contrast, ITPs have richer semantics than SMT solvers, enabling direct conversion between different types, including from finite fields to natural numbers and from natural numbers to bitvectors. Inspired by manual proofs \cite{succinct2025,avigad2022verified}, we rely on a $\toNat(\fe)$
function  (present in most ITPs \cite{lean}) that maps a finite field term to its underlying natural number value. Then, the $f(\fe)$ polynomial becomes a simple range check over the natural numbers: \toNat(\fe)  $\leq 2^N -1$. Similarly, instead of bit splitting, we can use $\toBV{N}{\nt}$, a function that converts a natural number term $\nt$ to a bitvector  of size $N$. Our \ffbv encoding is thus:
\begin{equation}
\ffbv_{\mathtt{ITP}}(\bv, \fe) \iff\; \toNat(\fe)<2^N \land \bv = \toBV{N}{\toNat(\fe)}.
\end{equation}
Decoding \ffbv in~\eqref{eq:corror1} yields:
\refstepcounter{equation}\label{eq:corror2}
\def\defaultHypSeparation{\hskip .01in}
\begin{prooftree}
\AxiomC{$
\begin{array}{@{}l@{}}
\hspace{.37cm}\bigwedge_{i \in \{0,\ldots,B-1\}}\;
\mathtt{bv1}[i] = \toBVOne{\toNat(x_i)}
\;\land\;
\toNat(x_i) \le 1
\\
\hspace{.37cm}\bigwedge_{i \in \{0,\ldots,B-1\}}\;
\mathtt{bv2}[i] = \toBVOne{\toNat(y_i)}
\;\land\;
\toNat(y_i) \le 1
\end{array}
$}
\alwaysSingleLine
\RightLabel{\hspace{.7cm}(\theequation)}
\UnaryInfC{$
\begin{array}{@{}r@{\;}c@{\;}l@{}}
(\text{\textbf{bvor}}\, \mathtt{bv1}\, \mathtt{bv2}) & = &
\toBVNoN\!\Bigl(B,\,
\toNat\!\Bigl(\sum_{i=0}^{B-1} 2^i (x_i + y_i - x_i \cdot y_i)\Bigr)\Bigr)
\\[2pt]
& \land &
\toNat\!\Bigl(\sum_{i=0}^{B-1} 2^i (x_i + y_i - x_i \cdot y_i)\Bigr)
\le 2^B - 1
\end{array}
$}
\end{prooftree}
While more compact, the ITP encoding raises the following question:
to discharge the two goals, how can one relate the conversions appearing in the goal to those applied to variables in the premises?

Our solution is an algorithm that soundly distributes the operators $\toNat$ and $\toBVNoN$ over finite field and natural number operations using range analysis. 
On the Jolt \texttt{OR} example, our approach, combined with bit-blasting, is able to automatically verify inputs with widths of up to 32 bits within five minutes and 8GB of memory.
Under the same limits, the \texttt{cvc5} SMT solver only verifies inputs with widths up to 9 bits under the weak encoding and  up to 3 under the full encoding.






 










\section{Background}%
\label{sec:bg}
In this section, we introduce the notation and logical frameworks used throughout the paper. We first describe the different arithmetic domains we will reason about~\cite{mceliece2012finite, dummit_and_foot}; we then discuss the logical setting we use~\cite{enderton2001mathematical}; and we conclude with an overview of Lean \cite{lean} and some existing Lean automation techniques~\cite{barrett2018satisfiability,  mohamed2025lean, bvdecide}. More details can be found in the cited work.

\subsection{Arithmetic Domains}
 We reason about expressions over three arithmetic domains: natural numbers, finite
fields, and bitvectors. We write $\NN$ for the natural numbers (including 0) and
$\NNOne$ for strictly positive naturals. 
In $\NN$, subtraction is truncated:
for $A,B \in \NN$, the term $A-B$ is defined as usual unless $B>A$, in which case it is $0$.
We use $ [i,j] $ to represent the set $ \{n\in\NN \mid i\le n \le j\} $ and $[n]$ to abbreviate $[1,n]$.

We use $\FFp$, with $P\in\NNOne$ prime, to denote a finite field of order $P$ (we consider only prime-ordered fields). The field $\FFp$ is isomorphic to the natural numbers modulo $P$. We use the function $\toNat : \FFp \rightarrow \NN$ to return the 
representative in \NN of an element of \FFp. 

We use $\BVN$ with $N\in \NNOne$ to denote a bitvector of width $N$. Bitvectors can be viewed as finite sequences of bits, and we write $b[i]$ to denote the $i$-th bit of a bitvector $b$. 
Bitvectors have a canonical natural number representation in the range $[0,2^N-1]$. We use the
function  $\bvToNat : \BVN \rightarrow \NN$ to map a bitvector to its natural number representation 
and the function $\toBVNoN : \NN \times \NN \rightarrow \BVN$ to map a natural number to its
$N$-bit encoding.

\subsection{Logical Setting \&  Notation }
\begin{table}[t]
\centering
\setlength{\tabcolsep}{5pt}
\renewcommand{\arraystretch}{1.25}

\begin{tabular}{ccc}
\toprule
Signature & Sort & Set of Symbols \\
\midrule

$\Sigma_{\FFp}$ & $\FFp$ &
$\{\cdot_{\FFp}, +_{\FFp}, -_{\FFp}\}$ \\

$\Sigma_{\NN}$ & $\NN$ &
$\{\cdot_{\NN}, +_{\NN}, -_{\NN}, \bmod_{\NN},
\leq_{\NN}, \geq_{\NN}\}$ \\


$\Sigma_{\BV}$ & $\BV_N$ &
\makecell[c]{$\{\cdot_{[N]}, +_{[N]}, -_{[N]}, \bmod_{[N]}
\leq_{[N]}, \geq_{[N]}\}$} \\
$\Sigma_{\FFp} \cup \Sigma_{\NN}$ & $\FFp, \NN$ &
$\{\toNat\}$ \\

$\Sigma_{\NN} \cup \Sigma_{\BV}$ & $\NN,  \BV_N$ &
$\{\toBVNoN_{[N]}, \bvToNat_{[N]}\}$ \\

\bottomrule
\end{tabular}

\vspace{0.25cm}
\caption{Relevant symbols from the theories referenced in this paper. For $\Sigma_{\BV}$, symbols are parameterized by $N \in \NNOne$, which represents the bit width.}
\label{tab:signature}
\vspace{-1cm}
\end{table}
\paragraph{Logic} We work in the logical setting
of many-sorted quantifier-free first-order logic with equality, using standard
terminology throughout.
We use $\Sigma$ to denote a many-sorted signature with a finite set of sorts,
including $\Bool$.  We also assume a family of equality symbols
$=_\sigma : \sigma \times \sigma \to \Bool$ for each $\sigma \in \Sigma$.
We assume the usual definitions of well-sorted terms and formulas (terms of sort $\Bool$).  We use $v,w$ to denote variables, $t$ to denote an arbitrary terms, and $\gamma$ to denote an arbitrary formula.
A theory consists of a signature together with a fixed interpretation of some subset of its
symbols.

We consider three base theories, each with a corresponding signature. The most relevant symbols are shown in Table~\ref{tab:signature}.
The finite field signature $\Sigma_{\FFp}$ includes the finite field sort, which we also denote $\FFp$, and the standard
finite field operators~\cite{hader2024smt}.
The natural number signature $\Sigma_{\NN}$ includes the natural number sort, which we also denote $\NN$, and the usual arithmetic
operators.
The bitvector signature $\Sigma_{\BV}$ includes, for each $N\in\NNOne$, the bitvector sort for bitvectors of size $N$, which we also denote $\BVN$ and the
SMT-LIB bitvector operators~\cite{smtlibbv}.
For all signatures, when the intended sorts of operators are clear from context, we 
write operators without explicit sort annotations.

We allow theory combination in the standard way.  When working in the combined theory, we additionally assume the signature is extended with three \emph{conversion operators} matching the three conversion functions ($\toNat$, $\toBVNoN$, and $\bvToNat$) mentioned above.  For convenience, we refer to both the functions themselves and the operators with the same notation.

\paragraph{Notation}
For a formula $\gamma$ and terms $s$ and $t$, we write
$s \preceq \gamma$ to denote that $s$ occurs as a subterm of
$\gamma$.
We write $\gamma[s/t]$ for the formula obtained by replacing
all occurrences of $s$ in $\gamma$ with $t$. We lift this notation pointwise to sets of formulas and tuples of sets of formulas.
For a set $S$ of formulas, $S[s/t]$ denotes the set obtained by
replacing all occurrences of $s$ in $S$ with $t$.
Similarly, for a tuple $T = (S_1, S_2)$ of sets of formulas, we define
$T[s/t] = (S_1[s/t], S_2[s/t])$. We also lift set membership to tuples of sets, using $\gamma\in T$ to denote $\gamma \in S_1 \cup S_2 $. We write $S \setminus \gamma$ to abbreviate
set $S \setminus \{\gamma\}$ and $S,\gamma$ to abbreviate $S\cup\{\gamma\}$. Finally, we write $S := S'$ to denote assignment of $S$ to
$S'$ and $S \models \gamma$ to indicate that $\gamma$ is logically implied by $S$.

\subsection{Lean \& SMT }
\paragraph{Lean} Lean~\cite{lean} is an interactive theorem prover (ITP) that allows users to
\emph{mechanically} verify mathematical theorems, including those about properties of software.
Lean provides a tactic language that enables users to interactively construct
proofs by transforming proof states, which consist of a goal together with a
local hypothesis context.
Lean's type checker 
automatically 
verifies proofs for users and throws errors during compilation if a proof is not valid.
This provides high confidence that the proof is correct, as it shrinks the trusted computing base (TCB) to Lean's small trusted kernel.
Lean is based on higher-order logic and is therefore more expressive than the
first-order setting described above. 
Our tactic operates on Lean goals but reasons in the restricted first-order setting, which is typically sufficient for ZKP arithmetizations. We thus use the notation of many-sorted first-order logic in this paper, with the understanding that we are working with a corresponding logical fragment within Lean.  

\paragraph{SMT}
 Satisfiability Modulo Theories (SMT) solvers allow users to \emph{automatically} verify theorems expressed in the language of many-sorted first-order logic with theories.  Some commonly supported theories include uninterpreted functions, integer and real arithmetic, and arrays and bitvectors~\cite{barrett2018satisfiability}.
SMT solvers are more automated but less expressive than ITPs.
To solve queries over multiple theories, SMT solvers employ dedicated theory solvers and theory combination mechanisms such as the Nelson-Oppen framework~\cite{nelson1979simplification}.
Since each theory solver operates independently, conversion operators are typically challenging to support and expensive for the solvers to reason about.

There are two main approaches for integrating SMT-like reasoning with Lean.
The first involves translating Lean proof
obligations into SMT queries, sending them to an external SMT solver, and then reconstructing the
proof in Lean to convince the kernel of the result \cite{mohamed2025lean}.
The reconstruction approach only supports a subset of SMT theories, not including (at the time of writing) finite fields and bitvectors. 
The second approach is to
implement SMT-like decision procedures directly in Lean, which keeps the reasoning inside
the proof assistant but may lead to increased solving times due to the need to perform the proof search within Lean.
In this paper, we follow the second approach and rely on \texttt{bv\_decide}, a
bit-blasting decision procedure implemented in Lean \cite{bvdecide}.

\section{BitModEq}%
\label{sec:tactic}
\begin{figure}[t]
\centering
\resizebox{\textwidth}{!}{%
\begin{tikzpicture}[
  type/.style={draw, rounded corners, dashed,minimum height=6mm, minimum width=7mm, align=center},
  proc/.style={draw, rounded corners, thick,  minimum height=6mm, minimum width=15mm, align=center},
  decision/.style={draw, rounded corners, minimum height=6mm, minimum width=5mm, align=center},
  arrow/.style={->, thick},
  ineq_1/.style={->, thick, dotted},
 ineq_2/.style={->, thick, dotted},
 ineq_3/.style={->, red},
  trl/.style={->, thick},
  node distance=1.0cm
]
\node[type] (FF) {$\Sigma_{\FFp} \cup \Sigma_{\mathbb{N}} \cup \Sigma_{\BV} $};
\node[proc, right=1.25 cm of FF] (TR1) {\texttt{to\_}$\NN$};
\node[type, right=1.25 cm of TR1] (N) {$\Sigma_\mathbb{N} \cup \Sigma_{\BV}$};

\node[proc, right=1.25 cm of N] (TR2) {\texttt{to\_}$\BV$};
\node[type, right=1.25 cm of TR2] (BV) {$\Sigma_{\BV}$};

\node[proc, below=1.2cm of FF] (RA0) {\texttt{rng\_analyze}};

\node[proc, below=3cm of TR1] (RA) {\texttt{rng\_analyze}};

\node[proc, below=3cm of TR2] (RA1) {\texttt{rng\_analyze}};
\node[proc, right=2cm of BV] (BVD) {\texttt{bit\_blast}};
\node[decision, below=3cm of BVD,  fill=green!70!black!70] (VALID) {\texttt{valid}~\ding{51}};
\node[decision, left=.9cm of VALID, fill=gray!30!] (CE) {\texttt{unknown \ding{55}}};
\node[decision, left=1.3cm of RA,  fill=gray!30!] (UNK) {\texttt{unknown} \ding{55} };

\draw[trl] (FF) -- (TR1);
\draw[trl] (TR1) -- (N);
\draw[trl] (N) -- (TR2);
\draw[trl] (TR2) -- (BV);
\draw[arrow] (BV) -- (BVD);

\draw[arrow] (BVD) -- (VALID);
\draw[arrow] (BVD) -- (CE);

\draw[ineq_1, loop above]
  (TR1) to
  node{\scriptsize $\Sigma_{\FFp} \cup \Sigma_\mathbb{N}$ ineqs. from \texttt{to\_}$\NN$}
  (TR1);

\draw[ineq_1]
  (TR1.south) to[bend right=50]
  node[pos=0.5, left]{\scriptsize $\Sigma_{\mathbb{N}}$ ineqs.}
  (RA.north);

\draw[ineq_3, bend right=50]
(RA.north) to
  node[midway, right]{\scriptsize diff. rule}
 (TR1.south);
 
\draw[ineq_2]
  (TR2.south) to[bend right=50]
  node[pos=0.5, left]{\scriptsize $\Sigma_{\mathbb{N}}$ ineqs. }
  (RA1.north);

\draw[ineq_3, bend right=40]
(RA1.north) to
  node[midway, right]{\scriptsize diff. rule}
 (TR2.south);
\draw[arrow] (RA0.south) --
 node[right]{\scriptsize
   \parbox{1cm}{\centering $\exists$ unsolved\\ineqs. }
 }
 (UNK.north);

\draw[arrow] (TR1.west) --
  node[pos=0.6, xshift=-25pt, align=center]
  {\scriptsize
    \parbox{1.8cm}{\centering
      $\Sigma_{\mathbb{N}}$ ineqs.\ from goal
    }
  }
(RA0.north);



\end{tikzpicture}
}
\caption{\BitModEq workflow. The tactic accepts terms over $\Sigma_{\FFp}$ and $\Sigma_{\BV}$, and terms originating from $\toNat$ and $\bvToNat$ operators. It first converts $\FFp$ terms to $\NN$ terms and then converts $\NN$ terms to $\BVN$ terms. Inequalities are discharged by \texttt{rng\_analyze}. If \texttt{rng\_analyze} cannot prove an inequality from the original goal, the tactic returns \texttt{unknown}; if it cannot prove a premise from a translation rule, a different rule is applied. 
Once only $\Sigma_{\BV}$ goals and hypotheses remain, the goals are discharged via bit-blasting. \texttt{valid} is returned if and only if bit-blasting solves the problem and no inequalities remain.}
\label{fig:workflow}
\vspace{-.4cm}
\end{figure}

In this section, we describe the \BitModEq algorithm. We begin with a high-level overview and then examine its core elements: conversion and range analysis.
\subsection{Overview}
Figure~\ref{fig:workflow} illustrates the \BitModEq workflow, which
operates in four phases:
\begin{enumerate}
\item Convert all terms of sort $\FFp$ in the proof context to terms of sort $\NN$.
\item Discharge any inequalities over $\NN$ in the original goal using range analysis.
\item Convert any remaining terms of sort $\NN$ to bitvector terms.
\item Discharge the final bitvector equality goal via the \texttt{bv\_decide} tactic.
\end{enumerate}
Since finite field elements do not admit a canonical bitvector representation,
we must select an intermediary domain.
Two natural candidates are the integers and the natural numbers.
We choose the natural numbers for two reasons.
Primarily, in the majority of the ZKP arithmetizations, finite field variables are used to simulate
bits and therefore do not represent negative values.
Second, converting integers to bitvectors requires reasoning about both lower
and upper bounds, whereas conversion from natural numbers requires reasoning
only about upper bounds, eliminating a significant portion of the proof obligations.

Both conversion passes, steps (1) and (3), traverse the input term from the root and push conversion
operators through arithmetic operators toward variables and constants.
If we can show that no modular wraparound is possible for an operand, we do not have to account for overflow or underflow and thus end up with much simpler resulting terms. 

For example, consider the \texttt{bvor} example from 
\Cref{sec:motivate}.
For $B=1$, when we are pushing \toNat{} inward in $\toNat(x_0 +y_0 - x_0\cdot{}y_0)$, if we can show that 
\begin{equation}
\toNat( x_0\cdot{}y_0) \leq \toNat(x_0 +y_0),
\label{eq:natineq}
\end{equation}
we know the expression won't become negative (which would result in wraparound) and thus can soundly drop the modulus, executing the transformation:  
\begin{equation}
\toNat(x_0 +y_0 - x_0\cdot y_0) \rightarrow \toNat(x_0 +y_0) - \toNat( x_0\cdot y_0).
\end{equation}
%
If we cannot prove Equation~\eqref{eq:natineq}, we have to account for possible wraparound, and our transformation becomes: 
\begin{equation} \toNat(x_0 +y_0 - x_0\cdot y_0) \rightarrow(\toNat(x_0 +y_0) + P - \toNat( x_0\cdot y_0))  \bmod P,
\end{equation}
which results in much larger bit widths that take longer to bit-blast. Therefore, we use different translation transformations (Section~\ref{sec:translation}) for cases where modular
wraparound must be accounted for and cases where it can be ruled out.
To prove the absence of wraparound, we rely on range analysis.

A natural design choice would have been to rely on Lean’s existing automation,
such as \texttt{ring\_nf}, \texttt{omega}, or \texttt{aesop}, to discharge
inequality goals.
However, initial experiments showed that these tactics are often overwhelmed by
large coefficients or struggle to handle variable dependencies and bound
information arising in ZKP arithmetizations.
We therefore design a dedicated range analysis algorithm.

Similar to the conversion passes, our range analysis, described in (Section~\ref{sec:range_analysis}),  is syntax-directed: it traverses
the formula from the root toward variables and constants, introducing symbolic
bounds for intermediate subterms.
Leaf goals are discharged using bounds from the hypothesis context,\footnote{All variables have bounds as they are either finite field or bitvector variables.} and
the resulting information is propagated back toward the root.
We additionally incorporate specialized patterns that capture common
bit-specific variable dependencies. While incomplete, our range analysis is largely independent of numeric magnitude
and can be more efficient than bit-blasting for large bit widths.
Thus, we also use range analysis to discharge $\Sigma_{\NN}$-inequalities
present in the original goals in step~(2).

The theorem is proven once range analysis discharges all inequalities and bit-blasting discharges the final bitvector equality. Otherwise, we return \texttt{unknown}.



\subsection{Translating from Finite Fields to Bitvectors}
\label{sec:translation}
\begin{figure}[t!]
 \fbox{%
  \begin{minipage}{0.97\textwidth}

    \centering
 \begin{center}
 \def\defaultHypSeparation{\hskip .01in}
  \AxiomC{  $t_1,t_2,t_3 : \FFp$}
   \AxiomC{$ \gamma \in \Gamma$}
   \AxiomC{$t_1-t_2+t_3 \preceq \gamma $}
  \LeftLabel{\mvZModSub}
  \TrinaryInfC{$\Gamma:= \Gamma[t_1-t_2+t_3/ (t_1+t_3) - t_2]$}
  \DisplayProof
\end{center}
\begin{center}
\def\defaultHypSeparation{\hskip .01in}

  \AxiomC{$t_1,t_2: \FFp$}
  \AxiomC{$t_1=t_2\in \Gamma$}
  \LeftLabel{\injNat}
  \BinaryInfC{$\Gamma:= \Gamma[ t_1=t_2 /\toNat(t_1) = \toNat(t_2)]$}
  \DisplayProof
\end{center}

\begin{center}
\def\defaultHypSeparation{\hskip .01in}
  \AxiomC{$t_1,t_2 : \FFp $ \hspace{.5em } $\gamma \in \Gamma$ \hspace{.5em} $\bowtie\ \in \{+, \cdot\}$  \hspace{.5em} $\toNat(t_1 \bowtie_{\FFp} t_2) \preceq \gamma $}
   \alwaysSingleLine
  \LeftLabel{\mvNat}
  \UnaryInfC{$ \Gamma:= \Gamma [ \toNat(t_1 \bowtie_{\FFp} t_2) / (\toNat(t_1)\bowtie_{\NN}  \toNat(t_2) ) \bmod P] $}
\DisplayProof
\end{center}
\begin{center}
\def\defaultHypSeparation{\hskip .01in}
  \AxiomC{ $t_1,t_2: \FFp$ \hspace{.5em} $t_3 : \mathtt{Bool}$ \hspace{.5em} $\gamma \in \Gamma$ \hspace{.5em} $\toNat(\mathtt{ite}( t_3,t_1, t_2))\preceq \gamma$  }
  \LeftLabel{\shortstack{\texttt{dist}\\\texttt{NatIte}}}
  \alwaysSingleLine
  \UnaryInfC{$
\Gamma:= \Gamma[
\toNat(\mathtt{ite}(t_3,t_1, t_2) / 
\mathtt{ite}(t_3, \toNat(t_1), \toNat( t_2))]
$}
\DisplayProof
\end{center}
\begin{center}
\def\defaultHypSeparation{\hskip .01in}
  \AxiomC{ $t_1,t_2: \FFp$  \hspace{.2em} $\gamma \in \Gamma$ \hspace{.2em} $\toNat(t_1- t_2) \preceq \gamma$ \hspace{.2em} $H \models \toNat(t_1) \geq \toNat(t_2)$  }
  \LeftLabel{\shortstack{\texttt{dist}\\\texttt{NatSub}}}
  \alwaysSingleLine
  \UnaryInfC{$\Gamma:= \Gamma[\toNat(t_1 - t_2) / \toNat(t_1)- \toNat(t_2)]$}
\DisplayProof
\end{center}
\begin{center}
\def\defaultHypSeparation{\hskip .01in}
  \AxiomC{ \hspace{.5em}$t_1,t_2 : \FFp$ \hspace{.5em} $\gamma \in \Gamma$}
  \alwaysNoLine
   \AxiomC{$\toNat(t_1- t_2) \preceq \gamma$}
\alwaysSingleLine
  \LeftLabel{\shortstack{\texttt{distNat}\\\texttt{SubOvrflw}}}
  \BinaryInfC{$\Gamma:= \Gamma[ \toNat(t_1  - t_2) / (\toNat(t_1) + P - \toNat(t_2) ) \bmod P ]$}
\DisplayProof
\end{center}
\begin{center}
\def\defaultHypSeparation{\hskip .01in}
  \AxiomC{
  $v_1 : \FFp$ \hspace{.3em}
  $v_1 \preceq \gamma$ \hspace{.3em} $\gamma\in\Gamma$ \hspace{.3em} $\toNat(v_1) \leq N \notin H $ for $N\in[0,P-1]$}
     \alwaysSingleLine
  \LeftLabel{\addBds}
  \UnaryInfC{$H:= H \cup \{ \toNat(v_1) \leq P-1 \})$}
\DisplayProof
\end{center}
\begin{center}
\def\defaultHypSeparation{\hskip .01in}
  \AxiomC{$t_1,t_2,t_3 : \NN$}
  \AxiomC{$\gamma \in \Gamma$}
   \AxiomC{$\bowtie \; \in \{+,\cdot\} $}
   \AxiomC{$(t_1  \bmod t_3 \bowtie t_2 \bmod t_3 ) \bmod t_3 \preceq \gamma$}
   
  \LeftLabel{\mvMod}
  \QuaternaryInfC{$\Gamma:= \Gamma[(t_1 \bmod t_3 \bowtie t_2 \bmod t_3) \bmod t_3 /  (t_1 \bowtie t_2) \bmod t_3$]}
\DisplayProof
\end{center}

\begin{center}
\def\defaultHypSeparation{\hskip .01in}
  \AxiomC{ $P\in \NNOne$ \hspace{.5em} $t_1 : \NN$}
  \AxiomC{$t_1 \bmod P \preceq \gamma$}
  \AxiomC{$\gamma \in \Gamma$}
  \AxiomC{$H \models t_1 \leq P-1$}
  \LeftLabel{\dropMod}
  \QuaternaryInfC{$\Gamma:= \Gamma [t_1 \bmod P /t_1]$}
\DisplayProof
\end{center}
\begin{center}
\def\defaultHypSeparation{\hskip .01in}
  \AxiomC{$t_1 : \FFp$}
   \AxiomC{$t_1 \cdot t_1 = t_1 \in \Gamma $}
  \LeftLabel{\sqrBds}
  \BinaryInfC{$\Gamma:= \Gamma[ t_1 \cdot t_1 = t_1 / \toNat(t_1) \leq 1] $}
\DisplayProof
\end{center}
\vspace{-1em}
\end{minipage}
}
\caption{Translation rules from $\FFp$  to $\NN$.  $\Gamma = \{ G,  H\}$ denotes the  proof context, where  $G$ is the set of goals and $H$ is the global set of shared hypotheses. 
}
\label{fig:tr1}
\vspace{-.5cm}
\end{figure}


\begin{figure}[t!]
    \centering

 \fbox{%
  \begin{minipage}{0.97\textwidth}

\begin{center}
\def\defaultHypSeparation{\hskip .01in}
  \AxiomC{ $N \in \NNOne$ \hspace{.4em} $t_1,t_2 : \NN$}
   \AxiomC{$t_1 = t_2 \in \Gamma $}
   \AxiomC{$H \models t_1 \leq 2^N -1  $\hspace{.4em}  $ H \models t_2 \leq 2^N-1$ }
  \LeftLabel{\injBitVec}
  \TrinaryInfC{$\Gamma:= \Gamma[ t_1=t_2 /\toBVN{t_1} = \toBVN{t_2}$]}
\DisplayProof
\end{center}

\begin{center}
\def\defaultHypSeparation{\hskip .01in}
  \AxiomC{ $N \in \NNOne$ \hspace{.4em} $t_1,t_2 : \NN$ \hspace{.4em} $\bowtie \; \in \{\leq, \geq\}$ \hspace{.4em} 
$t_1 \bowtie_{\NN} t_2 \in H $  }
  \alwaysNoLine
  \UnaryInfC{   $\toBVN{t_1} \bowtie_{[N]}  \toBVN{t_2} \notin H $  \hspace{.3em}
$H \models t_1 \leq 2^N -1$ \hspace{.3em} $ H \models t_2\leq 2^N-1$ }
\alwaysSingleLine
  \LeftLabel{\shortstack{\texttt{injBV} \\ \texttt{LeqHyp}}}
  \UnaryInfC{$H:=  H \cup \{ \toBVN{t_1} \bowtie_{[N]}  \toBVN{t_2}\}$}
\DisplayProof
\end{center}
\begin{center}
\def\defaultHypSeparation{\hskip .01in}
  \AxiomC{ $N \in \NNOne$ \hspace{.5em} $t_1,t_2 : \NN  $ \hspace{.5em}  $\gamma\in \Gamma$  \hspace{.5em}  $ \bowtie \; \in \{+_,\cdot\}$ \hspace{.5em}
$\toBVN{t_1 \bowtie_{\NN} t_2}) \preceq \gamma $ \hspace{.5em}}
  \LeftLabel{\mvBV}
  \alwaysSingleLine
  \UnaryInfC{$\Gamma:= \Gamma[\toBVN{t_1 \bowtie_{\NN} t_2}/ \toBVN{t_1}\bowtie_{[N]} \toBVN{t_2}]$}
\DisplayProof
\end{center}
\begin{center}
\def\defaultHypSeparation{\hskip .01in}
  \AxiomC{$N \in \NNOne$ \hspace{.5em} $t_1,t_2: \NN$ \hspace{.5em} $t_3 : \mathtt{Bool}$ \hspace{.5em} $\gamma \in \Gamma$ \hspace{.5em} $\toBVN{\mathtt{ite}(t_3, t_1, t_2)}\preceq \gamma$  }
  \LeftLabel{\shortstack{\texttt{distBV} \\ \texttt{Ite}}}
  \alwaysSingleLine
  \UnaryInfC{$
\Gamma:= \Gamma[
\toBVN{\mathtt{ite}(t_3 , t_1, t_2)} / 
\mathtt{ite}(t_3, \toBVN{t_1}, \toBVN{t_2})]
$}
\DisplayProof
\end{center}
\begin{center}
\def\defaultHypSeparation{\hskip .01in}
  \AxiomC{$N \in \NNOne$ \hspace{.5em} $t_1,t_2 : \NN$}
   \AxiomC{$\gamma \in \Gamma$}
   \AxiomC{$\toBVN{t_1-t_2} \preceq \gamma$}
  
      \alwaysNoLine
   \LeftLabel{\mvBVSub}
 \TrinaryInfC{ $H \models t_1 \le 2^N -1 $\hspace{.5em}  $
    H \models \toBVN{t_1} \geq \toBVN{t_2})$}
  \alwaysSingleLine 
  \UnaryInfC{ $\Gamma := \Gamma[\toBVN{t_1 - t_2} /\toBVN{t_1} - \toBVN{t_2} ]$
}

\DisplayProof
\end{center}
\begin{center}
\def\defaultHypSeparation{\hskip .01in}
  \AxiomC{ $C,N \in \NNOne$ \hspace{.5em} $t_1 : \NN$}
  \AxiomC{$\gamma \in \Gamma$}
   \AxiomC{$ C \leq 2^N -1 $}
   \AxiomC{$H \models t_1 \leq 2^N-1$}
   
   \LeftLabel{\mvBVMod}
 \QuaternaryInfC{ $\Gamma := \Gamma[\toBVN{t_1 \bmod C}  / \toBVN{t_1} \bmod \toBVN{C}]$ }

\DisplayProof
\end{center}

  \end{minipage}
  }
\caption{Translation rules from $\NN$  to $\BVN$.}
\label{fig:tr2}
\vspace{-.5cm}
\end{figure}
In this section, we present our \emph{translation calculus}, which contains two types of rules: rules that convert from finite fields to natural numbers (Figure~\ref{fig:tr1}) and rules that convert from natural numbers to bitvectors (Figure~\ref{fig:tr2}).

We follow the common practice of presenting rules that modify \emph{configurations}~\cite{pertseva2025integer, sheng2022reasoning}. In our calculus, a configuration is a proof context $\Gamma$, represented as a tuple $(G,H)$, where $G$ is a set of goals and $H$ is a set of global hypotheses shared across all goals. When a rule modifies only $G$ or only $H$ we omit the unchanged component.  Rules are presented in \emph{guarded assignment form}, where the premises describe the configurations to which the rule can be applied, and the conclusion describes the modification to the configuration. A sequence of applications is called a \emph{derivation}. We now present the two types of rules.

\paragraph{Finite Fields to Natural Numbers}\JP{We might want to highlight interesting rules if we run out of space}
\mvZModSub groups additions before subtractions to accumulate positive values and ensure minimal underflow  cases.
\injNat replaces finite field equations with natural number equations.
\mvNat distributes the $\toNat$ operator over addition and multiplication and adds $\bmod$ for soundness. Similarly \mvNatIte distributes the operator over if/then/else.

\mvNatSub distributes the $\toNat$ operator over subtraction when no negative values are possible. The $\bmod$ operator can be soundly dropped as $\toNat(t_1)$ ensures the starting expression is less than $P$ and subtraction by another natural number will only lower the final value. 
In cases where the premises of \mvNatSub cannot be discharged, we rely on \mvNatSubOverflow, which ensures soundness by eliminating the possibility of negative values by adding $P$ and  then ensuring the result is in bounds by applying $\bmod \; P$.
\addBds adds bounds for finite field variables present in the proof context to ensure equisatisfiability of the natural number constraints.
 \mvMod simplifies the term by factoring out $\bmod$.
When the term is smaller than the modulus, we apply \dropMod to remove the modulus. 

While the majority of ZKP arithmetizations only use finite field variables to represent bits, the restriction that a variable must be $0$ or $1$ relies on finite field reasoning (finite fields' not having zero divisors), because finite fields do not naturally support comparisons or the Boolean operator $\mathsf{or}$. Thus, we add \sqrBds  to detect bit encodings and translate them to natural number inequalities.

\paragraph{Natural Numbers to Bitvectors} 
As long as both sides of the expression fit into a bitvector of width $N$ without overflow, \injBitVec replaces natural number equalities with bitvector equalities of bit width $N$. While there might be inequalities, in our goal we leave them to range analysis and do not translate them. 
For hypothesis inequalities, we use \injBitVecLeq to maintain both the natural number and bitvector representations, so the inequalities can be used by the range analysis algorithm and by bit-blasting.
\mvBV distributes the $\toBVNoN$ operator over addition and multiplication.
Unlike \mvNat there is no need to add a $\bmod$ operator, as we remain in the $\bmod $ $2^N$ space. \mvBVIte distributes the $\toBVNoN$ operator over if/then/else. \mvBVSub distributes the $\toNat$ operator over subtraction when no negative values are possible. 
Unlike \mvNatSub, we add a second check that $t_1\leq 2^N -1$ because $t_1$ is not naturally a bitvector and might overflow, while in  \mvNatSub, $t_1$ is a $\ZMod$ and is thus constrained to be less than or equal to $P-1 $. We do not include a rule for cases where negative values are possible, since this case will never arise from the finite field to natural number translation rules. Last but not least, both \dropMod and \mvBVMod handle cases where the natural number encoding contains the $\bmod$  operator. 
 In cases where we cannot prove that the term is smaller than the modulus (\dropMod does not apply) we apply $\mvBVMod$ and distribute $\toBVNoN$ as long as both the modulus and the term are smaller than 2 to the power of the bit width.


\paragraph{Strategy}
\begin{figure}[t]
 \begin{subfigure}[t]{.48\textwidth}
    
\begin{algorithm}[H]
\KwIn{Proof context $\Gamma = (G,H)$}
\KwOut{Modified $\Gamma$}
\Fn{\texttt{to\_}$\NN(\Gamma)$}{
{progress $\leftarrow$ True \\}
\While{progress}{
progress $\leftarrow$ (\sqrBds || \injNat || \mvZModSub || \mvNat|| \mvNatIte || \mvNatSub|| \mvNatSubOverflow || \mvMod || \dropMod || \addBds) \\
}

}

\end{algorithm}
\caption{Strategy \toNN}
\label{fig:toNN}
\end{subfigure}
 \begin{subfigure}[t]{.51\textwidth}

\begin{algorithm}[H]
\KwIn{Proof context $\Gamma = (G,H)$}
\KwOut{Modified $\Gamma$}
\Fn{$\toBitVec(\Gamma)$}{
\DontPrintSemicolon
\LinesNumbered

{$b \leftarrow$ 2 \\}
\For{$\gamma \in \Gamma$} {
\For{$t\preceq\gamma$}
{$b\leftarrow$ \texttt{max}(\texttt{clcBVWdth}($t$),$b$)\\ }
}

{progress $\leftarrow$ True \\}
\While{progress}{
progress $\leftarrow$ $\injBitVecLeq(b)$
}
{progress $\leftarrow$ True \\}
\While{progress}{
progress $\leftarrow$ ($\injBitVec(b)$|| \mvBV|| \mvBVIte ||  \mvBVSub|| \mvBVMod) 
}

}
\end{algorithm}

\caption{Strategy \toBitVec }
\label{fig:toBV}
\end{subfigure}

\caption{Strategies for applying translation rules from Figures~\ref{fig:tr1} and~\ref{fig:tr2}.
 }
\label{fig:main_alg}
\vspace{-.5cm}
\end{figure}
Figure~\ref{fig:main_alg} presents our strategy, shown as pseudocode in which we assume that each translation rule returns \True iff it applies and changes $\Gamma$. We also assume that  \injBitVec and \injBitVecLeq take as input the target bit width. $(s_1\parallel \dots \parallel s_k)$ is evaluated from left to right, stopping after the first \True result. We first execute calculus rules according to Strategy $\toNN$ and then Strategy $\toBitVec$. $\calcBitWidth(t)$ is a function that overapproximates the maximum possible bit width needed to ensure $t$ can be represented without overflow.\footnote{In our implementation, the bit width is calculated per variable, but we use a global bit width here to keep the presentation simple.}
\paragraph{-} \toNN The strategy consists of a while loop, where a rule is only attempted once the previous rule no longer applies. We begin with \sqrBds to translate any bit encodings to bounds in case other rules rely on them. We then apply \injNat and use \mvZModSub to minimize underflow cases. Then, while one of \mvNat, \mvNatIte, \mvNatSub, or \mvNatSubOverflow applies, we
continue pushing $\toNat$ inward. We always attempt \mvNatSub before \mvNatSubOverflow to ensure the final translation is as simple as possible. Once we reach saturation, we factor out
the $\bmod$ operations that came from distribution. Once no more $\bmod$ operations can be factored we attempt to drop all $\bmod$ operators (line~4). 

\paragraph{-}\toBitVec We begin by setting the maximum bit width $b$ for the conversion to $2$ (line~2). Then, for each hypothesis and goal we update $b$ if any subterm requires a larger bit width (line~5). We calculate the maximum bit width across all of the terms to have a higher chance of proving the goal. 
For example, let $x\in \NN$, and suppose we translate a hypothesis $x \leq 1$ to $\toBVNoN(2,x) \leq 1_{[2]}$ but then use $\toBVNoN(4,x)$ in the goal. The goal then no longer follows from the hypothesis.  A possible counterexample could include $x = 3_{[4]}$ as $ (3_{[4]})_{[2]} = 2_{[1]}$. However, the value $3$ would not have been permitted by the original hypothesis. 

Upon establishing a maximum bit width, we use it to convert all natural numbers to bitvectors. We first apply \injBitVecLeq to saturation (to avoid applying it again if the converted inequalities are later modified, line 8),
then cast goals to bitvectors using \injBitVec, and then push \toBVNoN inward until saturation (line~11). 



We state soundness and termination results below.%
\footnote{Due to space constraints, the proofs of these and later theorems have been moved to the appendix, which also includes a discussion of strategy limitations and tradeoffs.}

\begin{theorem}
    Termination: For any context $\Gamma$, no infinite derivation is possible using strategies $\toNN$ and  $\toBitVec$.
    
\end{theorem}


\begin{theorem}
    Soundness: For any context $\Gamma$,  the bitvector translation $\Gamma'$ produced by $\toNN$ and $\toBitVec$ from $\Gamma$ is valid if and only if $\Gamma'$ is valid. 
\end{theorem}

\subsection{Range Analysis}
\label{sec:range_analysis}






\begin{figure}[t]
 \fbox{%
  \begin{minipage}{0.972\textwidth}
\begin{center}
 \def\defaultHypSeparation{\hskip .01in}
 \AxiomC{ $t_1, t_2  :\NN$ \hspace{.5em} $\bowtie\ \in \{\leq ,\geq\}$  \hspace{.5em} $ \gamma \in G $  \hspace{.5em} $ \HasVars(t_1 \bowtie t_2) $    \hspace{.5em}$t_1 \bowtie t_2  \equiv \gamma $} 
 \LeftLabel{\introMVar}
\UnaryInfC{$G:= (G \setminus \gamma) \cup \{t_1 \bowtie w ,\hspace{.5em} t_2\ ( \{\leq ,\geq\} \setminus \bowtie)\  w \}$ }
  \DisplayProof
\end{center}

\begin{center}
\def\defaultHypSeparation{\hskip .01in}

  \AxiomC{$t_1,t_2 : \NN  $ \hspace{.5em} $\bowtie\; \in \{ +,  \cdot\}$  \hspace{.5em} $\gamma \in G$ \hspace{.5em} $\HasVars(t_1 \bowtie t_2)$ \hspace{.5em}   $t_1 \bowtie t_2 \leq w \equiv \gamma $ }
   \alwaysSingleLine
  \LeftLabel{\ineqAddMul}
  \UnaryInfC{$G:=(G \setminus \gamma) \cup \{  t_1 \leq w_1,  t_2 \leq w_2,  w_1 \bowtie w_2 \leq w \}$ }
  \DisplayProof
\end{center}
\begin{center}
\def\defaultHypSeparation{\hskip .01in}
\AxiomC{ $t_1,t_2 : \NN $ }
  \AxiomC{ $\gamma \in G $   \hspace{.5em} $ \HasVars(t_1-t_2) $   }
  \AxiomC{    $t_1 - t_2 \leq w \equiv \gamma$
   }
 
  \LeftLabel{\leSub}
  \TrinaryInfC{$G:= (G \setminus \gamma) \cup \{  t_1  \leq w \} $}
  \DisplayProof
\end{center}
 
\begin{center}
\def\defaultHypSeparation{\hskip .01in}
\AxiomC{  $t_1, t_2 : \NN$  \hspace{.05em} $\gamma \in G $   \hspace{.05em} $ \HasVars(t_1 \bmod  t_2 )  $   \hspace{.05em}  $t_1 \bmod  t_2 \bowtie w \equiv \gamma $ }
  
  \LeftLabel{\ineqMod}
  \UnaryInfC{$G:= (G \setminus \gamma) \cup \{ t_2 \leq w  \} $ }
  \DisplayProof
\end{center}

\begin{center}
\def\defaultHypSeparation{\hskip .01in}
\AxiomC{ $t_3 : \BB$ \hspace{.1em} $t_1, t_2 :\NN $  \hspace{.1em}  $\gamma \in G $ \hspace{.1em}  $ \HasVars(t_1)$ \hspace{.1em}  $\HasVars(t_2)  $ \hspace{.1em}  $  \mathtt{ite}(t_3,  t_1,  t_2)\leq w  \equiv \gamma$   }
  \alwaysSingleLine
  \LeftLabel{\shortstack{\texttt{leq}\\\texttt{If}}}
  \UnaryInfC{$G:= (G \setminus \gamma) \cup  \{ t_1 \leq w_1,   t_2 \leq w_2 ,   \mathtt{max}(w_1, w_2) \leq w  \}  $  }
  \DisplayProof
\end{center}

\end{minipage}
}
\caption{Range analysis rules for inequality decomposition. We use $w$ and $w_i$ for placeholder variables and assume that any new placeholder variables in the conclusion are fresh. Other notation follows Figure~\ref{fig:tr1}.}
\label{fig:ineq_decomp}
\vspace{-.3cm}
\end{figure}

\begin{figure}[t]
 \fbox{%
  \begin{minipage}{0.97\textwidth}

\begin{center}
\def\defaultHypSeparation{\hskip .01in}
\alwaysNoLine
\AxiomC{$C \in \NN$ \hspace{.5em} $t_1 : \NN$ \hspace{.5em}   $\bowtie\; \in \{ \leq,  \geq\}$} \UnaryInfC{ $\gamma\in G$ \hspace{.5em} $\HasVars(t_1)$ \hspace{.5em} 
$ t_1 \bowtie  w \equiv \gamma$ \hspace{.5em} $ t_1\bowtie C \in H$ }
\alwaysSingleLine 
  \LeftLabel{\ineqHyp}
  \UnaryInfC{ $
  G :=
 ( G \setminus \{\gamma' | \gamma' \in G  \; \land \;  w \preceq \gamma' \}) \cup 
  \{\gamma'[w / C] \mid \gamma' \in G\setminus \gamma \}$
}
  \DisplayProof
\end{center}
\begin{center}
\def\defaultHypSeparation{\hskip .001in}
\AxiomC{ $v : \FFp$  \hspace{.5em} $v \notin \pVars $ }
  \AxiomC{$ \toNat(v) \leq w \equiv \gamma$}
  \AxiomC{$\gamma \in G$}

  \LeftLabel{\leZMod}
  \TrinaryInfC{$G :=
  (  G \setminus \{\gamma' | \gamma' \in G  \; \land \;  w \preceq \gamma' \}) \cup 
  \{\gamma'[w / P -1] \mid \gamma' \in G \setminus \gamma \}$}
  \DisplayProof
\end{center}
\begin{center}
\def\defaultHypSeparation{\hskip .01in}
\AxiomC{ $N \in \NNOne$ \hspace{.3em} $v : \BV_{[N]}$}
  \AxiomC{  $v \notin \pVars $ \hspace{.2em}   $\bvToNat(v) \leq w \equiv \gamma$}
  \AxiomC{$\gamma \in G$ }
  \LeftLabel{\leBV}
  \TrinaryInfC{$G :=
    (G \setminus \{\gamma' | \gamma' \in G  \; \land \;  w \preceq \gamma' \}) \cup 
  \{\gamma'[w / 2^N-1] \mid \gamma' \in G \setminus \gamma \} $}
  \DisplayProof
\end{center}

\begin{center}
\def\defaultHypSeparation{\hskip .01in}

\AxiomC{$t_1 : \NN$ \hspace{.5em} $\HasVars(t_1)$ \hspace{.5em} 
$t_1 \geq w \equiv \gamma$ \hspace{.5em} $\gamma \in G$ }
\alwaysSingleLine
  \LeftLabel{\geNat}
  \UnaryInfC{$G:=  (G \setminus \{\gamma' | \gamma' \in G  \; \land \;  w \preceq \gamma' \} ) \cup \{\gamma'[w / 0] \mid \gamma' \in G \setminus \gamma  \} $}
  \DisplayProof
\end{center}

\begin{center}
\def\defaultHypSeparation{\hskip .01in}

\AxiomC{$C \in \NN$ \hspace{.5em} $\bowtie  \;\in \{\leq, \geq\}$}  
\AxiomC{ $\gamma \in G$ \hspace{.5em} $C \bowtie w \equiv\gamma$ }
\alwaysNoLine
\BinaryInfC{$\forall\, \gamma' \in G. (w \not\preceq \gamma'  \vee \forall\, v \preceq \gamma'.\: v \notin \ogVars) $ }
\alwaysSingleLine
  \LeftLabel{\ineqConst}
  \UnaryInfC{$G:= ( G \setminus \{\gamma' | \gamma' \in G  \; \land \;  w \preceq \gamma' \} ) \cup 
  \{\gamma'[w / C] \mid \gamma' \in G  \setminus \gamma \; \} $}
  \DisplayProof
\end{center}

\end{minipage}
}

\caption{Range analysis rules for placeholder variable elimination. Notation follows that of Figure~\ref{fig:ineq_decomp}.
} 
\label{fig:meta_elim}
\end{figure}

\begin{figure}[t!]
 \fbox{%
  \begin{minipage}{0.97\textwidth}

 \begin{center}
 \def\defaultHypSeparation{\hskip .01in}
 \AxiomC{$C_1,C_2 \in \NN$}
  \AxiomC{$\gamma \in \Gamma$}
 \AxiomC{ $C_1 \leq C_2 \equiv \gamma$}

    \LeftLabel{\decide}
   \TrinaryInfC{ $G: = (G \setminus \gamma) \cup \{ \decide (\gamma) \}$}
  \LeftLabel{\rmvMVar}

  \DisplayProof 
\end{center}

\begin{center}
\def\defaultHypSeparation{\hskip .01in}
\alwaysNoLine
  \AxiomC{$ t_1,t_2, t_3, v_1,v_2: \NN $ \hspace{.5em}  $\bowtie\; \in \{\leq, \geq\}$ }
  \UnaryInfC{ $\gamma \in G$ \hspace{.5em} $t_1\bowtie t_2 \equiv \gamma$ \hspace{.5em} $ v_1,v_2 \preceq \gamma $ \hspace{.5em} $\OnlyVars(\gamma)$}


  \UnaryInfC{ $\HasSub(\gamma) \; \vee \; \HasVars(t_1) \land  \HasVars(t_2)$   \hspace{.5em} $  v_1 \leq 1 \in H$ \hspace{.5em}  $v_2 \leq 1 \in H$ }
  \alwaysSingleLine
  \LeftLabel{\shortstack{\texttt{ineq}\\\texttt{Cases}}}
 \UnaryInfC{$G:= (G \setminus \gamma) \cup \{  \gamma[v_1,v_0/ 0,0] ,  \gamma[v_1,v_0/1,0]  , \gamma[v_1,v_0/ 0,1]  , \gamma[v_1,v_0/ 1,1] \} $}
  \DisplayProof
\end{center} 
\begin{center}
\def\defaultHypSeparation{\hskip .01in}
\alwaysNoLine
  \AxiomC{ $N \in \NNOne$ \hspace{.5em} $ v_1, t_1, t_2 : \NN $ \hspace{.5em}  $\bowtie\; \in \{ \leq,  \geq\}$ }
  \UnaryInfC{$ \gamma \in G $ \hspace{.5em}$  v_1\cdot_{\NN}t_1 + (1-v_1)\cdot_{\NN}t_2 \bowtie w \equiv \gamma $  \hspace{.5em} \hspace{.5em}  $ v_1 \leq 1 \in H $ 
  }
  \LeftLabel{\ineqMux}
\alwaysSingleLine
  \UnaryInfC{$G:= (G \setminus \gamma ) \cup \{\mathtt{ite}(v_1 = 0 ,  t_2 , t_1) \bowtie w \}  $}

  \DisplayProof
  \end{center}

\end{minipage}
}

\caption{Range analysis rules for constant and bit reasoning. Notation follows that of Figure~\ref{fig:ineq_decomp}.} 

\label{fig:bit_constants}
\vspace{-.3cm}
\end{figure}

We now describe the range analysis algorithm that handles $\Sigma_{\NN}$ inequalities. The algorithm is presented as a calculus in 
 Figures~\ref{fig:ineq_decomp}-\ref{fig:bit_constants}.
At a high level, the algorithm introduces new placeholder variables and applies lemmas to rewrite goals to formulas that can be inferred from the hypotheses. It also detects cases where variable dependencies influence these bounds and enumeration is feasible, and discharges the resulting goals by enumeration. We first show how the algorithm works on a simple modulus removal goal that arises during translation of the example from Section~\ref{sec:motivate} and then explain the rules in more detail.  
\paragraph{Motivating Example Revisited} Recall that our hypothesis state includes the formulas $\{ \toNat(x_0) \leq 1,\ \toNat(y_0) \leq 1 \}$. One goal that arises when processing the conclusion is:
$\toNat(x_0) + \toNat(y_0)\leq P- 1$.
To prove such goals, we introduce placeholder variables, which are implicitly existentially quantified. We then obtain formulas that can be unified with hypotheses by instantiating the placeholder variables.  The evolution of the goal state is shown below.
\[
\begin{aligned}
G_1 &=  \{\, \toNat(x_0)+\toNat(y_0) \le w_1,\; P-1 \ge w_1  \,\} \\
G_2 &=  \{\, 
\toNat(x_0)\le w_2,\;
\toNat(y_0)\le w_3,\;
w_2+w_3\le w_1, \;  P-1 \ge w_1 \} \\
G_3 &=   \{\, 1+1\le w_1, P-1 \ge w_1 \} \\
G_4 &=  \{\, 1+1 \le P-1 \,\}
\end{aligned}
\]

The transition from $G_2$ to $G_3$ uses the hypotheses in $H$ to match the pattern in the goal, resulting in the concrete bounds obtained in $G_3$.

To automate the inequality reasoning we formulate a calculus made up of three types of rules: Rules that decompose inequalities using placeholder variables (Figure~\ref{fig:ineq_decomp}), rules that remove placeholder variables (Figure~\ref{fig:meta_elim}), and rules that reason about bits and constants (Figure~\ref{fig:bit_constants}). 

To help describe the rules, we introduce additional notation and helper functions. We use \emph{original variables}, denoted by the set  \ogVars, to describe variables that exist in the proof context \emph{before} starting range analysis, and \emph{placeholder variables}, denoted by the set \pVars, to describe \emph{implicitly existentially quantified} variables introduced by the range analysis.  $\HasVars(t)$ is true iff all variables occurring in $t$ (of which there must also be at least one) are not in $\pVars$.
$\OnlyVars(t)$ is true iff a term contains exactly 2 distinct variables not in $\pVars$.
$\HasSub(t)$ is true iff $t$ contains at least one subtract ($-_{\NN}$) operator.
\paragraph{Inequality Decomposition}  \introMVar introduces a placeholder variable into any inequality that contains at least one original variable and no placeholder variables. We ensure that placeholder variables are on the right-hand side of inequalities when original variables or constants are present in the expression, to make decomposition and elimination easier. \ineqAddMul and \leqIf rely on standard interval arithmetic \cite{hickey2001interval} to introduce additional placeholder variables into the context.  \leSub strengthens the goal by ignoring $t_2$, which is reasonable since in practice $t_2$ rarely has a lower bound other than zero. 
Finally, \ineqMod similarly strengthens the goal by keeping only the modulus. 

\paragraph{Placeholder Variable Elimination}
After decomposition, we eliminate placeholder variables by instantiating them with bounds derived from the context. The rule \ineqHyp instantiates placeholder variables using inequalities present in the hypotheses. The rule \leZMod uses the fact that finite field values are bounded from above by the field modulus, \leBV uses the fact that $N$-bit bitvectors are bounded above by $2^{N}-1$, and \geNat relies on the fact that natural numbers are always greater than or equal to zero. If a placeholder variable occurs in an inequality with a constant, we instantiate the placeholder variable with the constant itself using \ineqConst, but only if all other occurrences of the placeholder variable do not include original variables. The check ensures we do not get into a rule application cycle.

\paragraph{Constant and Bit Reasoning}
Once all placeholder and original variables are eliminated from a goal, we can use \eval to check if the derived constant term is true. Since the inequality decomposition rules are intentionally syntax-directed and do not account for dependencies between variables, we add rules that detect and resolve common dependency patterns in ZKP arithmetizations. \ineqCases identifies potential dependence between two bit-simulating variables $v_0$ and $v_1$. We restrict this rule to the setting when only two variables are present to keep the case analysis manageable. We conservatively flag a dependency when either ($i$) the goal contains a subtraction operator, or ($ii$) both sides of the inequality contain variables. When a potential dependency is detected, \ineqCases performs explicit case analysis over the four possible assignments to $v_0$ and $v_1$. \ineqMux handles encodings that simulate XORs; if a bit-simulating variable is used as guard, the goal can be turned into an if/then/else expression. 

Due to its inability to detect and solve all variable dependencies and the fact that our range analysis sometimes strengthen goal, it is not complete (e.g., \eval might return \False for a valid inequality) but it is terminating and sound. We state the theorems below and include the proofs in Appendix~\ref{ap:proof-range} due to space constraints. 

\begin{theorem}
    Termination: For any context $\Gamma$, there is no infinite derivation from the range analysis calculus.
\end{theorem}

\begin{theorem}
    Soundness: For any context $\Gamma$, if a derivation from $\Gamma$ reaches a context where the set of goals is empty, then all goals in $\Gamma$ hold.
\end{theorem}

\subsection{Implementation}
\label{sec:implementation}

We implement  \BitModEq in Lean v4.23.0. We use \bvdecide~\cite{bvdecide} for bit-blasting and \texttt{mathlib}~\cite{10.1145/3372885.3373824} for key lemmas.
%
The majority of our translation and range analysis rules are implemented using existing lemmas in \texttt{mathlib}. We proved additional lemmas for the rules: \mvNatIte, \mvNatSubOverflow, \addBds, \sqrBds, \injBitVecLeq, \mvBVIte, \mvBVSub, \mvBVMod, \leSub, and \ineqMux (\Cref{fig:tr1,fig:tr2,fig:ineq_decomp,fig:bit_constants}). 
To reliably pattern match and apply lemmas, we require formulas to be in a normalized form. We rely on Lean's \texttt{simp},  \texttt{split\_ifs}, and \texttt{split\_ands} to simplify the syntax of formulas and add rewrites for commutativity, associativity and the negation operator.
When doing range analysis, we always first try to apply \ineqCases, \ineqMux, \ineqHyp and \ineqConst to have a better chance of proving the goal. \geNat is only used when no other rules apply. 

%
We rely on Lean's metavariable mechanism (introduced as part of lemma hypothesis) to represent placeholder variables and on Lean's kernel to eliminate them throughout the proof context when applying placeholder elimination rules (Figure ~\ref{fig:meta_elim}).
The pipeline of \BitModEq can involve proving similar inequalities (e.g., when $P\geq 2^N$, one way to prove $t_1 \bmod P \leq 2^N$ is to prove $t_1 \leq P \land t_1 \leq 2^N$). Thus, to prove the soundness of the bitvector conversion step, when $P \geq 2^N$ we prove $t_1 \leq 2^N$ and use it to solve both goals.
\section{Evaluation}%
\label{sec:evaluation}

\begin{table}[t]
\centering
\setlength{\tabcolsep}{4pt}
\small

\begin{subtable}[t]{0.48\linewidth}
\centering
\small
\begin{tabular}{lcc}
\toprule

Solver & Lasso & T\&S \\
\midrule
\BitModEq & \textbf{24} & \textbf{12} \\
\texttt{cvc5} {\tiny weak} & 21 & 5 \\
\texttt{cvc5}{-split} {\tiny weak} & 10 & 2 \\
\texttt{cvc5} {\tiny strong} & 4 & 0 \\
\texttt{cvc5}{-split} {\tiny strong} & 3 & 0 \\
\midrule
\# Bench. & 24 & 15 \\
\bottomrule
\end{tabular}
\subcaption{Solved instances}
\label{tab:tab:jolt_results}
\end{subtable}
\hfill
\begin{subtable}[t]{0.48\linewidth}
\centering
\small
\begin{tabular}{lcc}
\toprule
Solver & Lasso & T\&S \\
\midrule
\BitModEq & 0 & 0 \\
SMT  {\tiny weak} & 0 & 0 \\
SMT  {\tiny strong} & 9 & 15 \\
\midrule
\# Bench. & 24 & 15 \\
\bottomrule
\end{tabular}
\subcaption{Statement generation timeouts
}
\label{tab:to-jolt}
\end{subtable}
\caption{\BitModEq and \texttt{cvc5} comparison for Lasso and T\&S benchmarks.}
\label{tab:results_jolt_all}
\vspace{-.2cm}
\end{table}


%

\begin{figure}[t]
  \centering
  \begin{subfigure}[t]{0.48\textwidth}
    \centering
    \includegraphics[width=\textwidth]{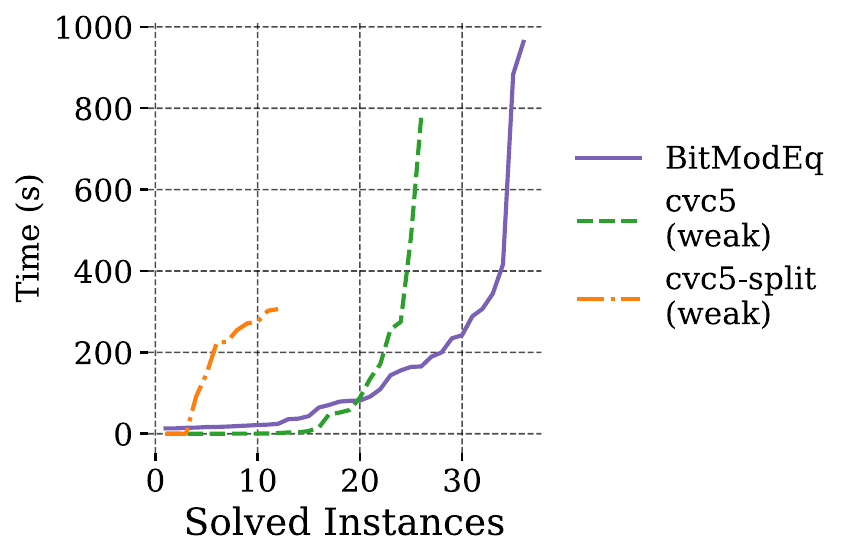}
    \caption{Jolt (Lasso and T\&S)}
    \label{fig:cactus-jolt}
  \end{subfigure}\hfill
  \begin{subfigure}[t]{0.48\textwidth}
    \centering
    \includegraphics[width=\textwidth]{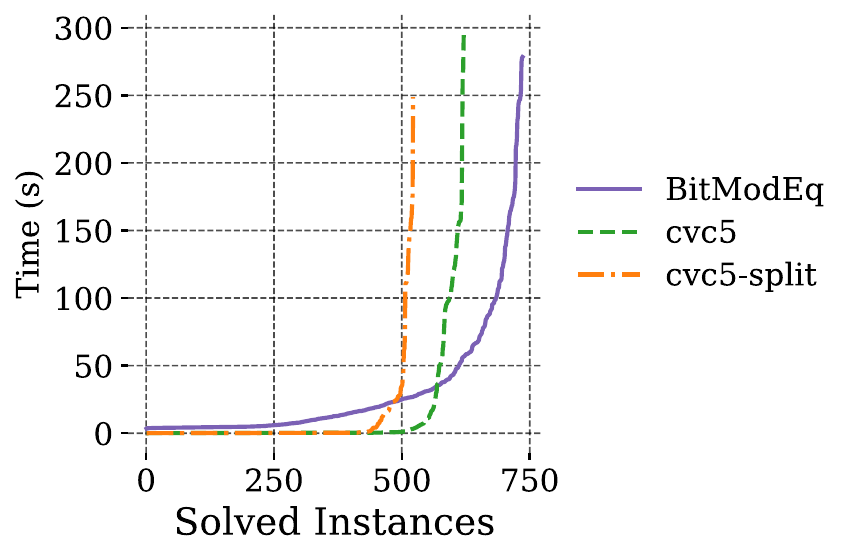}
    \caption{Sampled CirC }
    \label{fig:cactus-circ}
  \end{subfigure}
  \caption{Cactus plots across benchmark suites comparing \BitModEq and \texttt{cvc5}}
  \label{fig:cactus-side-by-side}
\vspace{-.4cm}
\end{figure}
\begin{table}[t]
\centering
\setlength{\tabcolsep}{2.5pt}
\renewcommand{\arraystretch}{1.15}
\begin{tabular}{l *{17}{c} c}
\toprule
Solver &
\multicolumn{17}{c}{Bit width} & All \\
\cmidrule(lr){2-18}
& 0 & 2 & 4 & 6 & 8 & 10 & 12 & 14 & 16 & 18 & 20 & 22 & 24 & 26 & 28 & 30 & 32 & \\
\midrule

\BitModEq
& 34 & 50 & 51 & 50 & \textbf{49} & \textbf{50} & \textbf{47} & \textbf{43} & \textbf{41} & \textbf{42} & \textbf{40} & \textbf{42} & \textbf{44} & \textbf{38} & \textbf{41} & \textbf{40} & \textbf{35} & \textbf{737}\\

\texttt{cvc5}
& \textbf{40} & \textbf{78} & \textbf{71} & \textbf{54}
& 38 & 35 & 33 & 34 & 30 & 27 & 26 & 25 & 27 & 26 & 27 & 27 & 23 & 621 \\

\texttt{cvc5}-split
& \textbf{40} & 74 & 61 & 47 & 32 & 31 & 28 & 27 & 22 & 21 & 20 & 20 & 21 & 20 & 20 & 21 & 17 & 522 \\
\midrule  
\# Bench. & 40 & 80 & 80 & 80 & 80 & 80 & 80 & 80 & 80 & 80 & 80 & 80 & 80 & 80 & 80 & 80 & 78 & 1318 \\
\bottomrule
\end{tabular}






\vspace{.3cm}
\caption{ \BitModEq and \texttt{cvc5} comparison on  CirC benchmarks (0--6 arguments; 0--32 bits), sampling up to two benchmarks per even bit width and operation (i.e., $\min(2,n)$ where $n$ is the number of benchmarks available). Extraction checks are not generated at 32 bits, as this is the base width for the operator.}
\label{tab:results_bits}
\end{table}

 
\begin{table}[t]
\centering
\setlength{\tabcolsep}{4pt}
\small

\centering
\small
\begin{tabular}{lccc}
\toprule
Solver  & Jolt & T\&S & CirC\\
\midrule
\BitModEq w/ full range analysis & \textbf{24} & \textbf{12} & \textbf{737} \\
\BitModEq w/o case splits & 21  & 8  & 720  \\  
\BitModEq w/ lean tactics & 7 & 4 &  434 \\

\# Bench. & 24 & 15 & 1318 \\
\bottomrule
\end{tabular}
\vspace{.3cm}
\caption{Ablation study of the \BitModEq range analysis procedure and comparison against existing Lean automation tactics.}
\label{tab:ablation_results}
\vspace{-1cm}
\end{table}
\begin{table}[t]
\centering
\small
\setlength{\tabcolsep}{10pt}
\renewcommand{\arraystretch}{1.1}
\begin{tabular}{lc}
\toprule
\textbf{\BitModEq Stage} & \textbf{Runtime (\%)} \\
\midrule
Range analysis & 52 \\
Translation & 23 \\
\texttt{bv\_decide} & 21 \\
Lean kernel checking & 4 \\
\bottomrule
\end{tabular}
\vspace{.3cm}
\caption{Average runtime breakdown of \BitModEq on solved Jolt benchmarks. We use type checking as a proxy for Lean kernel checking.}
\label{tab:runtime-breakdown}
\vspace{-1cm}
\end{table}

Our experiments answer two key empirical questions: 
(1) How does \BitModEq compare to the state of the art on ZKP arithmetization benchmarks?
(2) How critical is the proposed range analysis algorithm to \BitModEq{}'s effectiveness?   All experiments are
run on a cluster with Intel Xeon E5-2637 v4 CPUs.  

\subsection{Benchmarks}
To evaluate the effectiveness of \BitModEq, we verify ZKP arithmetizations extracted from two real-world systems described below. The evaluated systems are representative of the types of arithmetizations commonly used in many modern ZKP frameworks, including Cheesecloth and  Ligetron  \cite{cuellar2023cheesecloth, cuellar2025cheesecloth, wang2024ligetron}.
\paragraph{Jolt} Jolt~\cite{arun2024jolt} is a zkVM that relies on lookup tables to encode RISC-V instructions as finite field polynomials. The original Jolt design uses Lasso~\cite{setty2024unlocking} to perform 8-bit input lookup tables, while subsequent work, Twist\&Shout~\cite{cryptoeprint:2025/105}, extends this approach to 32-bit input tables. 
We extract the lookup table arithmetizations from both systems and verify their semantic equivalence to corresponding Lean bitvector operations.\footnote{We omit three Twist\&Shout instructions which depend on external assumptions. } We include a breakdown of benchmark size and complexity in Appendix ~\ref{ap:stats}.

\paragraph{CirC}
CirC~\cite{circ} is a zkVM compiler whose arithmetization passes were previously verified using \texttt{cvc5}~\cite{ozdemir2025bounded}. We implement a Lean backend for arithmetization verification in CirC and generate benchmarks at larger bit-widths than those considered in prior work.

\subsection{State of the Art Comparison}

\paragraph{Baselines}
Since there is no prior automation in Lean that works on ZKP arithmetizations, we compare \BitModEq against \texttt{cvc5}'s \cite{cvc5} two main finite field decision procedures, one based on full \gbs~\cite{CAV:OKTB23} and the other on split \gbs~\cite{CAV:OPBFBD24} (v.1.3.2), on benchmarks encoded in the SMT finite field and bitvector theories. 
For Jolt we evaluate the \emph{weak} and \emph{strong} SMT encodings, as discussed in Section~\ref{sec:motivate}. For CirC benchmarks, we rely on the existing SMT backend, where the overflow checks are checked by the compiler \cite{ozdemir2025bounded}.  In contrast, our Lean encoding always maintains the stronger guarantees. If generation of a benchmark for any encoding takes more than a minute we consider it a timeout.

\paragraph{Jolt} We run the Jolt benchmarks with a 20 minute timeout and 16GB memory limit. As summarized in Table~\ref{tab:tab:jolt_results}, \BitModEq solves three more benchmarks than the weak \texttt{cvc5} encoding for the 8-bit lookup tables, and seven more for the 32-bit lookup tables. 
Table~\ref{tab:to-jolt} shows that the strong SMT encoding times out during generation for 24 benchmarks, all of which involve range checks with large constants ($2^{32}$). The weak encoding does not timeout as it avoids the range check.

The cactus plot in Figure~\ref{fig:cactus-jolt} shows that while \texttt{cvc5} solves simpler benchmarks faster (e.g., \texttt{msb}, \texttt{lsb} arithmetization), \BitModEq significantly outperforms \texttt{cvc5} on larger benchmarks while also producing kernel-checked proofs, and despite operating under a stronger encoding.


\paragraph{CirC} Using the pipeline from prior work~\cite{ozdemir2025bounded}, we generate CirC benchmarks covering arithmetization passes with 0 to 6 arguments using bit widths from 0 to 32. Due to the large number of benchmarks, we focus on even bit widths, sample up to two benchmarks per arithmetization per bit width, and run with a five minute timeout and 8GB memory limit. As shown in Table~\ref{tab:results_bits}, \texttt{cvc5} outperforms \BitModEq on lower bit widths but scales poorly to the larger ones. \BitModEq  solves more benchmarks for 13 out of the 17 bit widths.

In general, \BitModEq performs well on constraints that simulate bit-level semantics with tightly bounded variables but is less effective on constraints with unbounded variables or heavy use of field-specific operations (e.g., inverses, negation). Using  large field constants, multiplicative inverse constraints
(e.g., equations of the form $x \cdot y = 1$), and field negation as an overapproximation, we estimate that 42\% of the CirC benchmarks
contain some  field-specific behavior.\footnote{In contrast, none of
the Jolt benchmarks rely on field-specific behavior.} These behaviors include arithmetizations corresponding to bitvector operators such as \texttt{bvnot} and \texttt{bvsle}. Although \texttt{cvc5}'s theory combination approach is more general and performs better when field-specific operations are present, its bit-splitting strategy introduces a large number of auxiliary variables, which significantly degrades performance at higher bit widths.

Across the benchmark suite, 130 SMT instances timeout during encoding generation. We include the breakdown in Appendix \ref{sec-appendix-circtimeout}. 
Similarly to the Jolt benchmarks, as seen in Figure~\ref{fig:cactus-circ}, \BitModEq is slower on simpler benchmarks but outperforms \texttt{cvc5} on larger ones.  


\subsection{Range Analysis Evaluation}
To evaluate the effectiveness of our range analysis  we substitute the algorithm for a portfolio of existing Lean automation tactics : \texttt{ring\_nf}, \texttt{linarith}, \texttt{nlinarith}, \texttt{aesop}, \texttt{grind}, and \texttt{omega}. 
 We additionally evaluate the impact of disabling case splitting (\ineqCases and \ineqMux).
Table~\ref{tab:ablation_results} shows the results. \BitModEq with the full range analysis solves substantially more benchmarks across all benchmark sets. Case splits play a more significant role for the Jolt benchmarks because they tend to contain more variable dependencies. 

\subsection{Runtime breakdown}
To understand the bottlenecks of \BitModEq, we measure the time spent in each stage of the pipeline on the solved Jolt benchmarks. Table~\ref{tab:runtime-breakdown} shows the results. Although essential for the success of \BitModEq, range analysis accounts for more than half of the total runtime. This is expected, since the pipeline attempts to prove range bounds for every operator encountered during translation.

\subsection{Discovered bugs}
During CirC evaluation, we identified a soundness bug in the 32-bit arithmetic right shifts (\texttt{bvashr}) related to integer overflow in CirC's Rust implementation. The issue was uncovered via SMT-based verification during tool comparison, did not occur at smaller bit widths, and shows the importance of verification beyond toy sizes.
The CirC developers have been notified and have fixed the issue.




\section{Related work}%
\label{sec:relwork}

Prior work on verifying ZKP arithmetizations broadly falls into two categories: manual verification in ITPs and automatic verification using SMT solvers. Manual approaches provide strong end-to-end correctness guarantees but require substantial proof engineering effort. For example, ~\cite{succinct2025} proves the correctness of the SP1 Hypercube zkVM instruction chips with respect to Sail's RISC-V semantics~\cite{armstrong2019isa}, and \cite{avigad2022verified} verifies that Cairo's AIR correctly constrains execution traces.
Automatic approaches aim to reduce manual effort by relying on SMT solvers. Özdemir et al.~\cite{ozdemir2025bounded} verify CirC’s finite-field-blasting pipeline using cvc5, checking all compiler arithmetizations from 0 to 4 arguments and 0 to 4 bits. While automated, this approach maintains large encodings and struggles to scale.

Beyond arithmetization pipelines, a large body of work focuses on the verification of ZK circuits themselves. On the proof-assistant side, Liu et al.~\cite{liu2024certifying} present CODA, a circuit language based on refinement types for specifying and proving circuit properties, while the CIVER tool~\cite{isabel2024scalable} defines transformation and deduction rules for verifying Circom circuits. On the SMT side, $QED^2$~\cite{pailoor2023automated} detects underconstrained circuits using finite field reasoning, and Zequal~\cite{stephens2025automated} checks consistency between circuits and witness generation. Complementary fuzzing-based approaches have also been effective at uncovering bugs in ZK circuits and compilers~\cite{chaliasos2025towards,takahashi2025zkfuzz,xiao2025mtzk}. While valuable for improving security, these techniques focus on circuit-level properties and do not establish end-to-end correctness guarantees. As demonstrated by bounded verification of CirC’s finite field blasting passes~\cite{ozdemir2025bounded}, correctness issues can occur at the level of arithmetization rules themselves, independent of the generated circuit.
Another closely related work is the verification of Lasso lookup tables by Kwan et al.~\cite{kwan2025verifying} in ACL2, which also relies on bit-blasting. Unlike our work, this work relies on ACL2's existing correspondence between finite fields and bitvectors and represents lookup tables as bits instead of finite field elements.  


Finally, our approach connects to prior work on SMT theory translations.  Zohar et al.~\cite{zohar2022bit} translate bitvector constraints to integers, while Pertseva et al.~\cite{pertseva2025integer} lift and lower constraints between integers and modular arithmetic using range analysis.  However, these works do not produce kernel-checked proofs in an ITP and do not simultaneously reason across finite field and bitvector theories, which is essential for verifying ZKP arithmetizations.

\section{Conclusion and Future Work}%
\label{sec:discuss}
This paper presents \BitModEq, a Lean tactic for automatically proving equivalence between finite field and bitvector encodings present in ZKP arithmetizations. By combining translation, range analysis, and bit-blasting in Lean, \BitModEq produces kernel-checked proofs and outperforms state-of-the-art solvers on real-world benchmarks. Multiple directions remain for future work. An important challenge is verifying systems along rarely exercised code paths, specifically addressing cases where bit widths exceed the modulus, which will require more scalable verification. It would also be valuable to characterize when translation to bitvectors is most effective and when direct theory combination or finite-field reasoning is preferable, as shown by CirC benchmarks. Finally, developing stronger native automation for finite field reasoning in Lean, including proof support for existing SMT-based decision procedures, remains crucial, as bitvector translation alone cannot capture all aspects of finite field semantics.



\bibliography{refs}

\appendix
\section{Proofs of Termination and Soundness of Translation}
\label{ap:proof-translation}
\setcounter{theorem}{0}
\begin{theorem}
    Termination: For any context $\Gamma$, no infinite derivation is possible using strategies $\toNN$ and  $\toBitVec$.
    
\end{theorem}

\begin{proof}
Since we apply $\toNN$ and  $\toBitVec$ only once, it is sufficient to prove that $\toNN$ and $\toBitVec$ contain no infinite cycles. We will prove this by mapping each of the rule applications in a while loop to a  strictly descending sequence.  \begin{enumerate}
 \item (Figure~\ref{fig:toNN} lines~3-4: Consider the tuple ordering (number of $\Sigma_{\FF}$-equalities in $\Gamma$, number of operators inside $\toNat$, number of $\bmod$ operators, number of $+_{\FFp}$ after $-_{\FFp}$ in all formulas, number of unbounded variables in $\Gamma$) \sqrBds and \injNat decrease the first entry.  \mvZModSub decreases the second to last entry without changing the first three. The next four rules decrease the second entry, without changing the first one. 
 The \dropMod and \mvMod decrease the third entry, without changing those before it. Last but not least, \addBds decreases the last by adding bounds. Since every entry is in this tuple is bounded below by 0, we will eventually reach termination.
 \item Figure~\ref{fig:toBV} lines~7-8: \injBitVecLeq decreases the number of $\Sigma_\NN$ hypotheses in $\Gamma$ without a corresponding $\Sigma_{\BV}$ hypothesis. 
 \item (Figure~\ref{fig:toBV}, lines~10-11) Consider the tuple ordering (number of $\Sigma_{\NN}$-equalities in $\Gamma$, number of operators inside $\toNat$). \injBitVec decreases the first entry, and all subsequent rules decrease the second without adding  $\Sigma_{\NN}$-terms.
 
\end{enumerate}
\end{proof}

\begin{theorem}
    Soundness: For any context $\Gamma$,  the bitvector translation $\Gamma'$ produced by $\toNN$ and $\toBitVec$ from $\Gamma$ is valid if and only if $\Gamma'$ is valid. 
\end{theorem}
\begin{proof}
    A careful analysis of the rules confirms that for each rule, the conclusion is valid (i.e. $H\models G$) iff the premise is.
    For each rewrite rule in $\toNN$ and $\toBitVec$, we have a proof of the corresponding lemma stating the equivalence of the premise and the conclusion in the Lean theorem prover.  These lemmas either already existed in Lean's Mathlib \cite{10.1145/3372885.3373824} or were proven by us as part of our implementation in Lean (see Section~\ref{sec:implementation}).
\end{proof}

\section{Proofs of Termination and Soundness of Range Analysis}
\label{ap:proof-range}

\begin{theorem}
    Termination: For any context $\Gamma$, there is no infinite derivation from the range analysis calculus.
\end{theorem}
\begin{proof}
Let $S$ be a sequence of rule applications.  We will show that any $S$ can be mapped to a strictly descending sequence in a well-founded order. We do this by showing how to construct a tuple from
a goal state $G$ in such a way that every rule application
reduces the tuple with respect to some well-founded order. We begin by defining helper functions. Let $\vars(G)$ denote the number of occurrences of original variables in all goals in $G$. Let $\ops(G)$ denote the number of natural number arithmetic operators $\{\cdot_{\NN}, +_{\NN}, -_{\NN},\bmod_{\NN}\}$ and $ \mathtt{ite}$ operators in all of the goals in $G$ that contain at least one original variable. Let $\expps(G)$ denote the number of goals in $G$ with an original variable but without a placeholder variable. Finally, let $|G|$ denote the number of goals in $G$. Our tuple ordering is $(\vars(G), \ops(G), |G|)$. We
define a total well-founded order $<_t$ on such tuples as the lexicographic order with non-negative standard integer ordering on all entries. We now show that every rule application decreases the tuple according to $<_t$. \introMVar decreases $\expps(G)$ by introducing placeholder variables. \ineqAddMul, \leSub,  \ineqMod, and \leqIf, all decrease $\ops(G)$ by splitting up and removing an operator from a goal that contains at least one original variable. They do this without increasing $\vars(G)$.  \ineqHyp, \leZMod, \leBV, and \geNat decrease $\vars(G)$, as they each remove a goal with at least one original variable from the proof context. \ineqConst decreases $|G|$ by removing a goal. 
\ineqCases decreases $\vars(G)$ by substituting constant values for the original variables in a goal. \ineqMux decreases $\ops(G)$ by replacing four arithmetic operators with one $\mathtt{ite}$ operator.
\end{proof}

\begin{theorem}
    Soundness: For any context $\Gamma$, if a derivation from $\Gamma$ reaches a context where the set of goals is empty, then all goals in $\Gamma$ hold.
\end{theorem}
\begin{proof}
Each rewrite rule in the range analysis calculus either maintains the logical equivalence of or strengthens the goals in the context.  This was verified by providing ``only-if'' Lean lemmas for each rule that either existed in Lean's Mathlib \cite{10.1145/3372885.3373824} or were proven in Lean during our implementation (Section~\ref{sec:implementation}).
\end{proof}

\section{Strategy Limitations}
\label{ap:limitation}

Applying \mvMod prior to \dropMod to saturation and then using an overapproximation in $\calcBitWidth$ can overinflate bit widths.
Consider cubing a variable $x \in \FF_7$. Under our current encoding,
\begin{equation}
x ^3  \rightarrow  (\toNat(x)^3 ) \bmod 7   \rightarrow  (\toBVNoN(9, \toNat(x))^3  \bmod  (7_{[9]})
\end{equation}
A more bit-width-aware encoding would keep the $\bmod$ in the second multiplication and propagate the modulus bit width in \calcBitWidth:
\begin{align}
& \rightarrow  (((\toNat(x)^2 ) \bmod 7) \cdot \toNat(x)) \bmod 7  \\ & \rightarrow  ((\toBVNoN(6, \toNat(x))^2  \bmod  (7_{[6]})) \cdot (\toBVNoN(6, \toNat(x)) \bmod  (7_{[6]})
\end{align}
However, each $\!\bmod\!$ operator requires additional range checks (to decide if either \dropMod or \mvBVMod  applies). Since overflow in the field is rare for our application, we consider the speedup gained from minimizing $\bmod$ operators worth the tradeoff, even if it results in some bit width inflation. 
\section{CirC Timeouts}
We present the timeouts for CirC SMT benchmark generations below. 


\label{sec-appendix-circtimeout}

\setlength{\tabcolsep}{2.5pt}
\renewcommand{\arraystretch}{1.15}
\begin{table}[]
\begin{tabular}{l *{17}{c} c}
\toprule
Solver &
\multicolumn{17}{c}{Bit width} & All \\
\cmidrule(lr){2-18}
& 0 & 2 & 4 & 6 & 8 & 10 & 12 & 14 & 16 & 18 & 20 & 22 & 24 & 26 & 28 & 30 & 32 & \\
\midrule

\BitModEq
& 0 & 0 & 0 & 0 & 0 &0 & 0 & 0 & 0 & 0 & 0 &0 & 0 & 0 & 0 & 0 & 0 &0  \\

SMT \tiny{CirC}
& 0 & 0 & 0 & 0 & 0 & 4 &3 & 4 & 1 & 9 & 21 & 18 & 20 & 18 &  18 & 12 & 0 & 130 \\

\midrule  
\# Bench. & 40 & 80 & 80 & 80 & 80 & 80 & 80 & 80 & 80 & 80 & 80 & 80 & 80 & 80 & 80 & 80 & 78 & 1318 \\
\bottomrule
\end{tabular}
\vspace{.5em}
\caption{CirC Statement generation timeouts after 1 minute}
\end{table}
\vspace{-.8cm}
It is worth noting that timeouts are generally concentrated at larger bit widths, but are absent  at bit width 32. Based on our inspection of CirC code, we hypothesize that this occurs because the preprocessing and folding passes are substantially faster for powers of 2.

\section{Benchmark Statistics}
\label{ap:stats}

We summarize the sizes and complexity of the CirC and Jolt  SMT benchmarks below.

\begin{table}[H]
\centering

\setlength{\tabcolsep}{4pt}
\renewcommand{\arraystretch}{1.1}

\begin{tabular}{lcccccccc}
\toprule
Benchmarks  &
Var. &
Var. &
Constr. &
Constr. &
Disj. &
Disj. &
Ops. &
Chars. \\
&
(avg) &
(max) &
(avg) &
(max) &
(avg) &
(max) &
(avg) &
(avg) \\
\midrule


SMT \tiny{CirC}
& 15.44
& 198
& 2.77
& 5
& 0.33
& 26
& 87
& 567{,}648\footnote{CirC benchmarks tend to have large variable names} \\
SMT \tiny{Jolt}
& 67.38
& 103
& 128
& 199
& 0
& 0
& 1013
& 49{,}218 \\

\bottomrule
\end{tabular}

\vspace{.3cm}
\caption{Statistics of the CirC and Jolt Lean and SMT benchmark suites.}
\end{table}



\end{document}